\documentclass[10pt]{article}

\usepackage{pifont}

\usepackage{natbib}

\usepackage[top=1.5in, left=0.78in, bottom=1in, right=0.78in]{geometry}

\usepackage{fleqn}

\usepackage{hyperref}

\usepackage{graphicx}

\usepackage{amsmath}

\usepackage{amsthm}

\usepackage{amsfonts}

\usepackage{amssymb}

\usepackage{natbib}

\usepackage[T1]{fontenc}

\usepackage{latexsym}

\usepackage{times}

\usepackage{setspace}

\usepackage[bf,footnotesize]{caption}

\usepackage[displaymath]{lineno}

\usepackage{color}

\usepackage{eurosym}

\usepackage{yhmath}

\usepackage{epstopdf} 

\usepackage{epsfig}

\usepackage[T1]{fontenc}

\usepackage{authblk}

\usepackage{graphicx}
\usepackage{amsfonts}
\usepackage{amsmath}
\usepackage{amssymb}
\usepackage{natbib}
\usepackage{color}
\usepackage{epstopdf} 
\usepackage{epsfig}
\usepackage[]{parskip}
\usepackage{multirow,array}
\usepackage{soul}

\usepackage{amsmath}               
  {
	 \newtheorem{theorem}{Theorem}
      \newtheorem{assumption}{Assumption}
      \newtheorem*{acknowledgements}{Acknowledgements}
      \newtheorem{property}{Property}
	 \newtheorem{definition}{Definition}
      \newtheorem{lemma}{Lemma}
      \newtheorem{example}{Example}
      \newtheorem{proposition}{Proposition}

  }

\title{Insights on the Theory of Robust Games}

\author[1]{G. P. Crespi\thanks{giovanni.crespi@uninsubria.it}}

\author[2]{D. Radi\thanks{davide.radi@unipi.it \ and \ davide.radi@vsb.cz}}

\author[1]{M. Rocca\thanks{matteo.rocca@uninsubria.it}}

\affil[1]{\emph{\small Department of Economics, Insubria University, Varese (VA), Italy.}}

\affil[2]{\emph{\small Department of Economics and Management, University of Pisa, Pisa (PI), Italy.}}
\affil[2]{\emph{\small  Faculty of Economics, V\u{S}B--Technical University of Ostrava, Ostrava, Czech Republic.}}

\date{}

\begin{document}

\maketitle

\begin{abstract}
Restricting the attention to static games, we consider the problem of ambiguity generated by uncertain values of players' payoff functions. Uncertainty is represented by a bounded set of possible realizations, and the level of uncertainty is parametrized in such a way that zero means no uncertainty, the so-called nominal counterpart game, and one the maximum uncertainty. Assuming that agents are ambiguity averse and they adopt the worst-case optimization approach to uncertainty, we employ the robust-optimization techniques to obtain a so-called robust game, see \cite{AghassiBertsimas2006}. A robust game is a distribution-free model to handle ambiguity in a conservative way. The equilibria of this game are called robust-optimization equilibria and the existence is guaranteed by standard regularity conditions. The paper investigates the sensitivity to the level of uncertainty of the equilibrium outputs of a robust game. Defining the opportunity cost of uncertainty as the extra profit that a player would obtain by reducing his level of uncertainty, keeping fixed the actions of the opponents, we prove that a robust-optimization equilibrium is an $\epsilon$-Nash equilibrium of the nominal counterpart game where the $\epsilon$-approximation measures the opportunity cost of uncertainty. Moreover, considering an $\epsilon$-Nash equilibrium of a nominal game, we prove that it is always possible to introduce uncertainty such that the $\epsilon$-Nash equilibrium is a robust-optimization equilibrium. Under some regularity conditions on the payoff functions, we also show that a robust-optimization equilibrium converges \emph{smoothly} towards a Nash equilibrium of the nominal counterpart game, when the level of uncertainty vanishes. Despite these analogies, the equilibrium outputs of a robust game may be qualitatively different from the ones of the nominal counterpart. An example shows that a robust Cournot duopoly model can admit multiple and asymmetric robust-optimization equilibria despite the nominal counterpart is a simple symmetric game with linear-quadratic payoff functions for which only a symmetric Nash equilibrium exists.

\medskip


\end{abstract}

\section{Introduction}
\label{intro}

Ambiguity affects the choice of economic agents, as demonstrated in his seminal contribution by \cite{Ellsberg1961}. As opposed to risk, where an objective probability distribution describes some possible occurrence,  ambiguity (also known as uncertainty or Knightian uncertainty, see \cite{Knight1921}) is characterized by the inability of the decision maker to formulate a unique probability distribution or by his/her lack of trust in any single probability estimate. The experimentally documented attitude to prefer situations with known probabilities to unknown ones, to the extent that these can be compared, is defined uncertainty aversion or ambiguity aversion, see \cite{Epstein1999}.

In game theory, uncertainty and risk are related to the actions of the opponents as well as to the environment of the game. A cornerstone model that describes the decision mechanisms of players in this context of incomplete information is the Bayesian-game framework proposed by Harsanyi, see, e.g., \cite{Harsanyi1967}, \cite{Harsanyi1968} and \cite{Harsanyi1968b}. It represents a game-theoretical model to describe how players handle risk in a non-cooperative strategic setting, with objective probabilities available about the possible realizations of the environment of the game and the actions of the opponents. Nevertheless, this model does not provide a behavior setting and an equilibrium concept when players are affected by ambiguity, which is essential, for example, in discussing the refinement of Nash equilibria or to infer about players' selection problem among multiple Nash equilibria, see, e.g., \cite{Marinacci2000}, and in mechanism design where a planner does not know agents' beliefs about other players' types as well as the environment of the game, see, e.g. \cite{CremerMcLean1985} and \cite{BergemannMorris2005}. As indicated in \cite{JehielMeyerterVehnMoldovanuZame2006}, the sensitivity of Bayesian-Nash equilibrium to agents' beliefs suggests the use of more robust (distribution-free) notions of equilibrium for incomplete-information games. A very conservative concept is the one of \emph{ex-post equilibrium}, which is a Nash equilibrium under all possible realizations of the uncertain parameters. Introduced in \cite{HolmstromMyerson1983} with the name of \emph{uniform incentive compatibility} and used for the first time in the context of auctions in \cite{CremerMcLean1985}, ex-post equilibrium is increasingly studied in game theory, see, e.g., \cite{Kalai2004}, and it is often used in mechanism design as a very robust solution concept, see, e.g., \cite{CremerMcLean1985} and \cite{PerryReny2002}.\footnote{In oligopoly games, for example, ex-post solutions are considered in \cite{KlempererMeyer1989} where firms handle uncertainty by committing themselves to use a specific supply function and determine prices, or quantities, once uncertainty vanishes. Concerning mechanism design, \cite{BergemannMorris2005} looked at mechanisms that implement the social choice correspondence in ex-post equilibrium. Robust mechanism design in income taxation and public goods is considered instead in \cite{Bierbrauer2009}, where an ex-post implementation is adopted in a large economy model of income taxation and public goods provision under the assumption that the social benefits from public goods provision are a priori unknown.}

The ex-post equilibrium concept presents several drawbacks. Despite \cite{Kalai2004} shows that an equilibrium of large games is an ex-post Nash equilibrium, this remains a very conservative distribution-free equilibrium concept and its existence depends on the structure of the game and on players' uncertainty. Moreover, a recent literature on interdependence value environment has obtained positive and negative results using this solution concept, see, e.g., \cite{DasguptaMaskin2000} and \cite{BergemannValimaki2002}. Therefore new equilibrium concepts may be useful. The modern decision theory offers many models to describe the behavior of agents under ambiguity aversion, see \cite{GilboaMarinacci2013} for a recent survey on the topic. The most famous ones are the Choquet expected utility model proposed in \cite{Schmeidler1989} and the maxmin expected utility model introduced in \cite{GilboaSchmeidler1989}. The Choquet expected utility model accommodates the different degrees of confidence in one's own probability assignments or beliefs that a decision maker can experience. The model employs non-additive probabilities to represent uncertainty aversion. The maxmin expected utility model, followed by \cite{MaccheroniMarinacciRustichini2006} and others, captures the idea that an agent has beliefs about the state of the world represented by a set of probability distribution functions and he/she maximizes with respect to the minimum expected value taken over the possible probabilities.

In light of these developments in decision theory, alternative approaches to uncertainty have been proposed in game theory. Considering uncertainty about opponent agents' actions, \cite{Marinacci2000} introduces the \emph{ambiguous games}, a modification of the normal form game that allows the presence of vagueness in players' beliefs over the opponents' choice of strategies. In ambiguous games, players' behavior is expressed through a Choquet expected utility model where non-additive probability measures reflect the ambiguity aversion of agents. The model is flexible and distinguishes pessimistic players that, in the presence of ambiguity, emphasize the lower payoffs, and optimistic players that instead emphasize the higher ones. Different is instead the uncertainty about the environment of a game, which implies partial knowledge by players of their own payoff function. In the simplest case the payoff function depends on some parameters, the values of which are unknown. A player may have some ambiguous knowledge about the probability distribution of the uncertain parameters, for example he/she knows that the parameters with unknown values are drawn according to a set of possible probability distributions, or he/she may know that the possible realization of the uncertain parameters are described by deterministic sets. In an attempt to describe the uncertainty aversion of agents, \cite{GilboaSchmeidler1989} axiomatized a choice model, the maxmin expected utility model mentioned above, which mirrors in the soft robust-optimization technique when uncertainty is represented by ambiguous knowledge about the probability distribution of the parameters with unknown values, see \cite{Ben-TalGhaouiNemirovski2009}, and in the particular case of robust-optimization when parameters with unknown values are described by deterministic sets, see \cite{BenTalNemirovski1998}. Adopting the soft robust-optimization approach, an agent maximizes his guaranteed expected payoff, a measure of robustness against distributional variation, while in the robust-optimization approach an agent maximizes its guaranteed payoff. The last one is useful to accommodate bounded uncertainty and it can be seen as a special case of the previous model, which can be employed also when the possible realizations of the parameters with unknown values are not bounded.

In a context of bounded-uncertainty about the possible realizations of the values of the parameters of the payoff functions, \cite{AghassiBertsimas2006} and \cite{CrespiRadiRocca2017} employ a robust-optimization approach to ambiguity in which the agents rely on deterministic-set representations of uncertainty instead of ambiguous probabilistic descriptions of the environment of the game. A game where agents are uncertain about the shape of their own payoff function, as well as the ones of their opponents, and have an ambiguity aversion expressed by the robust-optimization approach to uncertainty, is denoted \emph{robust game}. The robust-game framework offers a distribution-free equilibrium concept, known as \emph{robust-optimization equilibrium}. 
It is a distribution-free equilibrium concept and it is less conservative than the ex-post equilibrium concept, in fact the first one implies the second one but not the vice versa, see  \cite{AghassiBertsimas2006}. The advantages are an existence theorem, computational methods offered by the robust-optimization techniques developed in the last two decades, see, e.g., \cite{BenTalNemirovski1998} and \cite{Ben-TalGhaouiNemirovski2009}, and theoretical results showing the relations with other equilibrium concepts in game theory. Regarding this last aspect, in \cite{CrespiRadiRocca2017} it has been shown that a robust-optimization equilibrium is an $\epsilon$-Nash equilibrium of the nominal counterpart game, that is of the game without uncertainty.

The aim of the current paper is to measure the sensitivity of the equilibrium outputs of a robust game with respect to uncertainty. This is a relevant aspect in forecasting the result of a game and in policy analysis as players have a subjective perception and representation of uncertainty and an erroneous specification of this uncertainty is possible. In this respect, the contribution of the current paper is new and alternative to the previous ones that are mainly focused on studying the output of a game once the uncertainty is specified. To study the sensitivity to uncertainty, a new modeling framework is proposed where uncertainty is characterized by the shape of the uncertainty set and by the level of uncertainty. This representation of the uncertainty set is crucial to study the effect of the ambiguity aversion. By focusing on the sensitivity towards the level of uncertainty, the theoretical insights underline that the equilibrium output of a robust game is an $\epsilon$-Nash equilibrium of the nominal counterpart game. Moreover, the $\epsilon$ waiver of an extra payoff required to play an $\epsilon$-Nash equilibrium measures the opportunity cost of uncertainty, or cost of uncertainty aversion. Therefore, the robust optimization equilibria of a robust game are confined in the set of $\epsilon$-Nash equilibria of the nominal counterpart game. In addition, any $\epsilon$-Nash equilibrium can be seen as a robust-optimization equilibrium assuming that players are affected by uncertainty in a suitable way. The result provides a theoretical foundation for the $\epsilon$ waiver of an extra payoff required to play an $\epsilon$-Nash equilibrium in terms of opportunity cost of uncertainty. The theoretical foundation proposed has the advantage of offering a criterion of choice for $\epsilon$-Nash equilibria, in the sense that a specific characterization of the uncertainty set, in terms of shape and level of uncertainty, is usually consistent with only a subset of the set of all $\epsilon$-Nash equilibria of a nominal game. Under certain conditions on the payoff functions of a game, the sensitivity to the level of uncertainty is also characterized in terms of robust-optimization equilibria that converge smoothly towards Nash equilibria when uncertainty vanishes.

The equilibrium outputs of a robust game are not only sensitive to the level of uncertainty but also to the shape of the uncertainty set, which are the two ingredients that define the uncertainty of a player. In the spirit of Knight, see \cite{Knight1921}, the level of uncertainty can be interpreted as the amount of confidence that a player has on his/her knowledge of the true values of the parameters of the payoff function, which is the essential of ambiguity modeling. On the contrary, the shape of uncertainty set represents a further degree of freedom offered by the robust games that is not available in other modeling frameworks that account for ambiguity in a context of strategic interaction. This second degree of freedom makes robust games a flexible modeling framework, able to replicate complicated equilibrium configurations with very simple modeling assumptions. In fact, for certain shapes of the uncertainty sets and assuming a sufficiently high level of uncertainty, the equilibrium outputs of a robust game may be qualitatively different from the ones of the nominal counterpart game. Special equilibrium outputs can be obtained even by simple geometric configurations of the uncertainty sets. A robust version of the Cournot duopoly model as in \cite{SinghVives1984}, underlines that uncertainty, in the form of a simple uncertainty set, leads to multiple robust-optimization equilibria, while the nominal version of the game admits a unique Nash equilibrium. Moreover, the Nash equilibrium is symmetric, the firms produce the same output, while the set of robust-optimization equilibria includes also asymmetric equilibrium outputs. The multiple robust-optimization equilibria arise thanks to non monotonic decreasing best-reply functions, which is a feature that can be observed in a nominal game only by assuming complicated payoff functions. See the stream of research that sparked after the seminal contribution in \cite{Rand1978}. In this respect, the example underlines the modeling flexibility offered by robust games. 

The road map of the paper is as follows. Section \ref{RobustGames} introduces the basic concepts of robust game theory. Section \ref{ExistenceROE} provides existence results for robust-optimization equilibria. Section \ref{ThandFoundation} introduces the concept of opportunity cost of uncertainty and contains the main theoretical insights on robust games. A simple example is provided where a symmetric two-player game with linear-quadratic payoff function depending on a parameter with unknown value can generate multiple and asymmetric robust-optimization equilibria, where only one of the many converges smoothly to a Nash equilibrium (it has a Nash equilibrium counterpart). Section \ref{Application} introduces a robust version of a duopoly model similar to the one in \cite{SinghVives1984} and provides results about the existence of robust-optimization equilibria.
Section \ref{Conc} concludes and provides indications for future research directions in robust-game theory. \ref{Ap:T} contains the proofs of the theoretical results in Sections \ref{ExistenceROE} and \ref{ThandFoundation}. \ref{AlgorithmBR} contains an algorithm to compute the worst-case best reply functions when the uncertainty sets are of polyhedral shape. \ref{AppendixB} contains the proofs of the conditions for the existence of the robust-optimization equilibria of the robust duopoly model in Section \ref{Application}.

\section{Robust games and equilibria}\label{RobustGames}

\bigskip

In an incomplete-information game, it is possible to distinguish uncertainty about opponents, considered in the ambiguous games introduced in \cite{Marinacci2000}, and uncertainty about your own payoff function or payoff environment, see, e.g., \cite{AghassiBertsimas2006} and \cite{CrespiRadiRocca2017}. 
Considering finite-person, non-cooperative, simultaneous-move, one-shot games only, we focus on payoff uncertainty and we consider incomplete-information games where payoff functions depend on some parameters, the values of which are not known in advance. Players are aware of this uncertainty and they have perfect knowledge of the set of all possible realizations of the uncertain parameters, the so called \emph{uncertainty set}. Employing the information about the uncertainty set, each player chooses the action that maximizes his/her own maximum guaranteed payoff. Such a player is denoted robust player, and a game populated by robust players is called a robust game.

More formally, we assume a finite set $\mathcal{N}=\left\{1,2,\ldots,n\right\}$ of players and $A_{i}$ is the action space of player $i$. Set $A = \times_{i=1}^{n}A_{i}$, we define by $f_{i}: W^{\delta_{i}}_{i} \times A \rightarrow \mathbb{R}$ the payoff function of player $i$ which is known except for the value of a parameter vector $\boldsymbol{\alpha}_{i}\in W^{\delta_{i}}_{i}$. Denoted by $f_{i}\left(\boldsymbol{\alpha}_{i};x_{i},\mathbf{x}_{-i}\right)$, the payoff function of a robust player $i$ depends on the realization of the uncertain parameter vector $\boldsymbol{\alpha}_{i}\in W^{\delta_{i}}_{i}$, on the player $i$'s action $x_{i}\in A_{i}$ and on the action of his/her opponents $\mathbf{x}_{-i}\in A_{-i}=\times_{j\in \mathcal{N};j\neq i}A_{j}$. Here, $W^{\delta_{i}}_{i}$ is a deterministic uncertainty set depending on a parameter $\delta_{i}\in\left[0,1\right]$ and defined as $W^{\delta_{i}}_{i}=\delta_{i} U_{i}+\left(1-\delta_{i}\right)\boldsymbol{\alpha}^{0}_{i}$, where $U_{i}$ is a closed and bounded deterministic set representing all possible realizations of the vector parameter $\boldsymbol{\alpha}_{i}$, while $\boldsymbol{\alpha}^{0}_{i}\in U_{i}$ is a singleton. According to this setting and given an action profile $\left(x_{i},\mathbf{x}_{-i}\right)$, player $i$'s vagueness about payoff is represented by the set $\left\{\left . f_{i}\left(\boldsymbol{\alpha}_{i};x_{i},\mathbf{x}_{-i}\right) \right |\boldsymbol{\alpha}_{i} \in W^{\delta_{i}}_{i}\right\}$ and $\delta_{i}$ measures the level of uncertainty, which vanishes when $\delta_{i}=0$ and it is maximum when $\delta_{i} = 1$.\footnote{In a generalized version of this model $U_{i}$ may be a space of probability distributions representing vagueness of player $i$ and $\boldsymbol{\alpha}^{0}_{i}$ a single probability distribution (in case of absence of vagueness) from which the unknown parameters (parameters with unknown values) of the payoff function are drawn.} Denoting a game without uncertainty a \emph{nominal game}, the proposed representation of the uncertainty set allows to define a unique nominal counterpart of a robust game, i.e. the nominal game obtained when $\max_{i\in \mathcal{N}}\delta_{i}=0$.
This latter game is called the \emph{nominal counterpart} of the robust game.\footnote{Note that for each robust game defined as in Definition \ref{Def::RG} there is a unique nominal counterpart game, i.e. a game with same players, same action space and same payoff functions but no uncertainty. However, from a nominal game it is possible to define an infinite variety of robust games because infinite are the possibilities to define the uncertainty sets. Therefore, we avoid the term \emph{robust counterpart}.}

In a robust game, each player $i$ is ambiguity averse (or uncertainty averse) and determines his/her action $x_{i}$ by maximizing his/her \emph{worst-case payoff function}:
\begin{equation}\label{WCWPF}
\rho^{\delta_{i}}_{i}\left(x_{i},\mathbf{x}_{-i}\right)\triangleq \min_{\boldsymbol{\alpha}_{i}\in W^{\delta_{i}}_{i}}f_{i}\left(\boldsymbol{\alpha}_{i};x_{i},\mathbf{x}_{-i}\right)
\end{equation}
The actions that players undertake by maximizing their worst-case payoff functions are denoted \emph{robust-optimization strategies}. In case uncertainty vanishes, the robust-optimization strategies become the well-known Nash strategies.

A robust optimization strategy derives from an extreme form of ambiguity aversion, in the sense that any possible probability knowledge assigned to a player would not make an action profile less appealing to him/her than the worst-case approach to uncertainty. This extreme uncertainty attitude may appear paranoid in the sense that it is equivalent to assuming that \emph{nature} will choose a realization of the unknown values of the parameters as if to spite the player. This extreme feeling can be, however, relaxed in the modeling framework here proposed. By interpreting $\boldsymbol{\alpha}^{0}_{i}$ as the real realization of the vector of parameters and inspired by \cite{Knight1921}, the level of uncertainty of a player can be seen as the level of confidence that the player has on his/her knowledge of the true values of the parameters of the payoff function. Then, the ambiguity aversion of a player can be measured by his/her level of uncertainty, and despite the worst-case approach to uncertainty, we can assume that a player is not ambiguity averse by assuming a level of uncertainty equal to zero.


To underline this aspect, note that a worst-case payoff function depends on the level of uncertainty and has the following property.

\medskip

\begin{property}\label{Prop1}
The worst-case payoff functions are such that $\rho^{\delta_{i}^{1}}_{i}\left(x_{i},\mathbf{x}_{-i}\right)\geq \rho^{\delta_{i}^{2}}_{i}\left(x_{i},\mathbf{x}_{-i}\right)$, for all $\delta_{i}^{2}>\delta_{i}^{1}$, $\delta_{i}^{1},\delta_{i}^{2}\in\left[0,1\right]$, $\forall \left(x_{i},\mathbf{x}_{-i}\right)\in A$ and $\forall i\in \mathcal{N}$. 
\end{property}

\medskip

Assuming that the worst-case payoff functions are \emph{well defined}, i.e. the minimum of functions $f_{i}$ with respect to $\boldsymbol{\alpha}_{i}\in W^{\delta_{i}}_{i}$ exists and is finite for all $\left(x_{i},\mathbf{x}_{-i}\right)\in A$, a robust game is defined as follows.

\medskip

\begin{definition}\label{Def::RG}
Denoted by $\left\{ A_{i},f_{i},W_{i}^{\delta_{i}} : i \in \mathcal{N} \right\}$, a robust game is a normal form game $\left\{ A_{i},f_{i} : i \in \mathcal{N} \right\}$ and an $n$-tuple $\left\{W_{i}^{\delta_{i}}\right\}_{i=1}^{n}$ of uncertainty sets, such that player $i$ possible payoffs related to action $\left(x_{i},\mathbf{x}_{-i}\right)$ belong to $\left\{\left . f_{i}\left(\boldsymbol{\alpha}_{i};x_{i},\mathbf{x}_{-i}\right) \right |\boldsymbol{\alpha}_{i} \in W_{i}^{\delta_{i}}\right\}$ and he/she maximizes his/her own worst-case payoff function $\rho^{\delta_{i}}_{i}\left(x_{i},\mathbf{x}_{-i}\right)$.
\end{definition}

\medskip

According to this definition and by defining \emph{equivalent} two games that have the same players, the same action spaces and players' behavior is independent of which of the two games is considered, we have that a robust game is equivalent to a nominal game when for each player the worst-case payoff function of the robust game coincides with the payoff function of the nominal game. That is, the robust game $\left\{ A_{i},f_{i},W_{i}^{\delta_{i}} : i \in \mathcal{N}\right\}$ is equivalent to the nominal game $\left\{ A_{i},\rho^{\delta_{i}}_{i}: i \in \mathcal{N} \right\}$.\footnote{Note that for each robust game there is a unique equivalent nominal game, while for each nominal game there are an infinite variety of equivalent robust games as the same worst-case payoff function can be obtained by different combinations of payoff functions and uncertainty sets.} Using this correspondence between a robust game and a nominal game, \cite{AghassiBertsimas2006} and \cite{CrespiRadiRocca2017} introduce the following equilibrium notion for robust games.

\medskip

\begin{definition}\label{definition::ROE}
A robust-optimization equilibrium (ROE in short) of a robust game $\left\{ A_{i},f_{i},W_{i}^{\delta_{i}} : i \in \mathcal{N} \right\}$ is a Nash equilibrium of the game $\left\{ A_{i},\rho^{\delta_{i}}_{i}: i \in \mathcal{N} \right\}$.
\end{definition}

\medskip

\section{The existence of a robust-optimization equilibrium}\label{ExistenceROE}

The definition of robust-optimization equilibrium emphasizes the similarities with the Nash equilibria in many respects. Indeed, once the worst-case payoff functions are derived, searching for a robust-optimization equilibrium of a robust game is equivalent to searching for a Nash equilibrium. Therefore, the same techniques and algorithms can be used. Moreover, the problem of existence for a robust-optimization equilibrium is equivalent to the problem of existence for a Nash equilibrium. Specifically, when the usual properties (see, e.g., \cite{Nash1950}) imposed on the payoff functions for the existence of a Nash equilibrium in a nominal game are satisfied by the worst-case payoff functions of a robust game, then a robust-optimization equilibrium exists for that robust game.

\medskip

\begin{theorem}\label{EquilibriumExistenceROEandNASHequivalentGame}
Consider a finite-person, non-cooperative, simultaneous-move, one-shot robust game $\left\{ A_{i},f_{i},W_{i}^{\delta_{i}} : i \in \mathcal{N} \right\}$, in which there is no private information. 
Assume that $A_{i}$ is a non-empty, closed, bounded and convex subset of an Euclidean space for all $i\in\mathcal{N}$. Moreover, assume that the game is equivalent to the nominal game $\left\{ A_{i},\rho_{i}^{\delta_{i}} : i \in \mathcal{N} \right\}$ with the worst-case payoff functions that are continuous in $A_{i}\times A_{-i}$, and concave in $A_{i}$ for every $\mathbf{x}_{-i}\in A_{-i}$. Then, a robust-optimization equilibrium of the robust game $\left\{ A_{i},f_{i},W_{i}^{\delta_{i}} : i \in \mathcal{N} \right\}$ exists.
\end{theorem}

\medskip

The existence result in Theorem \ref{EquilibriumExistenceROEandNASHequivalentGame} is based on the assumption that the worst-case payoff functions satisfy the usual properties, that means continuity and concavity. Instead of considering the worst-case payoff functions, the conditions for the existence of a robust-optimization equilibrium can be imposed directly on the payoff functions that define a robust game.


\medskip

\begin{assumption}\label{Ass1}
Let $A_{i}$, with $i\in\mathcal{N}$, be subsets of Euclidean spaces. We assume that:
\begin{itemize}
\item $A_{i}$ is a non-empty, closed, bounded, and convex set, for all $i\in \mathcal{N}$;
\item $U_{i}\subset \mathbb{R}^{\nu_{i}}$ is a non-empty, closed, bounded, and convex set, for all $i\in \mathcal{N}$, where $\nu_{i}$ is the number of entries in vector $\boldsymbol{\alpha}_{i}$;
\item $f_{i}$ are continuous on $W^{\delta_{i}}_{i}\times A$;
\item $f_{i}$ are concave in $x_{i}$, $\forall\boldsymbol{\alpha}_{i}\in W^{\delta_{i}}_{i}$.
\end{itemize}
\end{assumption}

\medskip

These restrictions are imposed hereafter and they ensure that the worst-case payoff functions satisfy the conditions imposed in Theorem \ref{EquilibriumExistenceROEandNASHequivalentGame} for the existence of a robust-optimization equilibrium.

\medskip

\begin{lemma}\label{PropWorstCasePayoffFunctions}
Under the conditions imposed in Assumption \ref{Ass1}, for all $i\in\mathcal{N}$ the following properties hold true:
\begin{itemize}
\smallskip
\item[-] $\rho^{\delta_{i}}_{i}:A\rightarrow \mathbb{R}$ is a continuous function;
\smallskip
\item[-] $\rho^{\delta_{i}}_{i}\left(\cdot,\mathbf{x}_{-i}\right)$ is concave for every $\mathbf{x}_{-i}\in A_{-i}$.
\end{itemize}
\end{lemma}

\medskip

The results in Lemma \ref{PropWorstCasePayoffFunctions} indicate that the conditions imposed in Assumption \ref{Ass1} are a particular case of the more general sufficient conditions imposed in Theorem \ref{EquilibriumExistenceROEandNASHequivalentGame} for the existence of a robust-optimization equilibrium. However, the conditions imposed in Assumption \ref{Ass1} regard directly the elements that define a robust game. Therefore, they allow a direct comparison with the necessary conditions imposed by the Nash's Theorem, see \cite{Nash1950} (and \cite{Dutang2013} for a recent survey). The comparison underlines that the sufficient conditions to impose on the payoff function for the existence of a robust-optimization equilibrium are stronger than the ones required for the existence of a Nash equilibrium. In fact, the continuity of the payoff function is also required with respect to the uncertainty set, and the concavity is required for all possible realizations of the unknown values of the parameters.

The results in Lemma \ref{PropWorstCasePayoffFunctions}, also discussed in \cite{CrespiRadiRocca2017} but without a formal proof, indicate therefore that the restrictions in Assumption \ref{Ass1} are sufficient (but not necessary) to ensure the existence of a robust-optimization equilibrium as well as a Nash equilibrium for the nominal counterpart game as stated in the following theorem, the proof of which is in \ref{Ap:T}.

\medskip

\begin{theorem}\label{EquilibriumExistenceROE}
Consider a finite-person, non-cooperative, simultaneous-move, one-shot robust game $\left\{ A_{i},f_{i},W_{i}^{\delta_{i}} : i \in \mathcal{N} \right\}$, in which there is no private information. 
Under Assumption \ref{Ass1}, we have that:
\begin{itemize}
\smallskip
\item[-] The robust game has at least a robust-optimization equilibrium;
\smallskip
\item[-] All robust games $\left\{ A_{i},f_{i},W_{i}^{\delta^{+}_{i}} : i \in \mathcal{N} \right\}$, with $\delta^{+}_{i}\in\left[0,\delta_{i}\right)$, have a robust-optimization equilibrium.
\end{itemize}
\end{theorem}

\medskip

The second point in Theorem \ref{EquilibriumExistenceROE} indicates that if a robust game has a robust-optimization equilibrium because the sufficient conditions imposed in Assumption \ref{Ass1} are satisfied, then reducing the level of uncertainty, the existence of at least a robust-optimization equilibrium, therefore the existence of the Nash equilibrium of the nominal counterpart game, are guaranteed. This result does not imply that a robust-optimization equilibrium for a robust game ensures the existence of at least a Nash equilibrium for the nominal counterpart game, or the existence of at least a robust-optimization equilibrium when the level of uncertainty shrinks. In fact, the existence of a robust-optimization equilibrium for a robust game that does not satisfy the conditions in Assumption \ref{Ass1} is also possible, and in this case the nominal counterpart game may not have a Nash equilibrium.

A further remark on the existence of a robust-optimization equilibrium is required. Due to the ambiguity aversion of players, specifically due to their robust or worst-case attitude towards uncertainty, some possible realizations of the uncertain parameters may be irrelevant for the robust game. To be precise, let us define the \emph{worst-case frontier} of an uncertainty set $W_{i}^{\delta_{i}}$ as that set of values of the unknown parameters that are the worst-case realization of the uncertainty set for at least a given action profile of players:
\begin{equation}
\left\{\boldsymbol{\alpha}^{*}_{i} \ | \  \exists \left(x_{i},\mathbf{x}_{-i}\right) \in A \quad \text{s.t.} \quad \boldsymbol{\alpha}^{*}_{i}=\arg\min_{\boldsymbol{\alpha}_{i}\in W^{\delta_{i}}_{i}}f_{i}\left(\boldsymbol{\alpha}_{i}; x_{i},\mathbf{x}_{-i}\right) \right\}
\end{equation}
Then, it is obvious that two robust games that differ for the uncertainty sets only, but the uncertainty sets have the same worst-case frontiers, are equivalent games according to the definition of equivalent games introduced above. The consequence is that the two robust games will have the same robust-optimization equilibria and the players will behave exactly in the same way, whatever the action of the opponents. Moreover, for the existence of a robust-optimization equilibrium it is sufficient that the conditions in Assumption \ref{Ass1} hold true for the worst-case frontiers of the uncertainty sets only. In other words, the knowledge of the worst-case frontier is the only information on an uncertainty set which is relevant to characterize a robust game.

\section{Foundations and theoretical insights}\label{ThandFoundation}

Uncertainty implies a loss of profit for player $i$ which determines the so-called \emph{opportunity cost of uncertainty} for player $i$.\footnote{Since in a robust game players are averse to uncertainty and the loss of profit caused by uncertainty depends on this attitude, the opportunity cost of uncertainty can be interpreted as the cost of aversion to uncertainty.} Given the actions of player $i$'s opponents, the opportunity cost of uncertainty for players $i$ measures the extra profit that player $i$ gains if the uncertainty about his/her payoff function vanishes. Therefore, given $\mathbf{x}_{-i}$ the vector of actions of player $i$'s opponents and let $\delta_{i}$ be the level of uncertainty for player $i$, the opportunity cost of uncertainty for player $i$ is
\begin{equation}\label{Ebm}
C^{\delta_{i}}_{i}\left(\mathbf{x}_{-i}\right):=\max_{x_{i}\in A_{i}}\rho^{0}_{i}\left(x_{i},\mathbf{x}_{-i}\right)-\rho^{0}_{i}\left(x^{+}_{i}\left(\delta_{i}\right),\mathbf{x}_{-i}\right)\text{; }\quad  \text{where} \quad x^{+}_{i}\left(\delta_{i}\right)\in \arg\max_{x_{i}}\rho^{\delta_{i}}_{i}\left(x_{i},\mathbf{x}_{-i}\right)
\end{equation}

According to the definition, for each given set of actions of the opponents, player $i$'s opportunity cost of uncertainty is the difference between the payoff player $i$ would obtain by maximizing his/her nominal payoff function and the value of the nominal payoff function in correspondence of the action that maximizes his/her worst-case payoff function. Consider a robust game where payoff uncertainty involves player $i$ only, which has a robust-optimization equilibrium and the nominal counterpart game admits a Nash equilibrium. In the case in which the opponents play their Nash equilibrium strategies and these strategies do not depend on the action of player $i$, the opportunity cost of uncertainty for player $i$ represents the loss that he/she undergoes when a robust-optimization equilibrium is played instead of a Nash equilibrium of the nominal counterpart game. More generally, the opportunity cost of uncertainty is not the extra payoff player $i$ obtains in the nominal counterpart game. In fact, the Nash-equilibrium strategies of player $i$'s opponents are not independent of player $i$'s action and, therefore, of player $i$'s payoff uncertainty. Moreover, in a robust game where payoff uncertainty involves also player $i$'s opponents, player $i$'s opponents play their robust-optimization strategies rather than their Nash-equilibrium strategies. Therefore, the opportunity cost of uncertainty for player $i$ is the extra profit he/she obtains when the actions of the opponents are fixed and his/her uncertainty vanishes.\footnote{Since the actions of the opponents depend on their level of uncertainty, we have that player $i$'s opportunity cost of uncertainty depends implicitly on the level of uncertainty of the opponents.}

\subsection{Uncertainty aversion leads to play an $\epsilon$-Nash equilibrium}

The opportunity cost of uncertainty does not represent, therefore, the players' loss of profit when a robust-optimization equilibrium is played instead of a Nash equilibrium of the nominal counterpart game. This notwithstanding, it links a robust-optimization equilibrium of a robust game to an $\epsilon$-Nash equilibrium of the nominal counterpart game. 

\medskip

\begin{definition}
The action profile $\left(x_{1}^{*},\ldots,x_{n}^{*}\right)\in A$ is an $\epsilon$-Nash equilibrium of game $\left\{ A_{i},f_{i}: i \in \mathcal{N} \right\}$, when for each $i\in \mathcal{N}$
\begin{equation}
f_{i}\left(\boldsymbol{\alpha}^{0}_{i};x^{*}_{i},\mathbf{x}^{*}_{-i}\right)\geq f_{i}\left(\boldsymbol{\alpha}^{0}_{i};x_{i},\mathbf{x}^{*}_{-i}\right)-\epsilon \quad \forall x_{i}\in A_{i}
\end{equation}
\end{definition}

\medskip

The basic idea behind the notion of $\epsilon$-Nash equilibrium is that a player accepts to play a strategy that is not optimal with respect to Nash definition, yet he/she will not deviate unless the payoff improvement is greater than $\epsilon$. The following theorem (see proof in \ref{Ap:T}) underlines that a robust-optimization equilibrium is an $\epsilon$-Nash equilibrium of the nominal counterpart game, where $\epsilon$ is the maximum of players' opportunity costs of uncertainty.

\medskip

\begin{theorem}\label{Theorem::DeltaROEm}
If $\left(x_{1}^{*},\ldots,x_{n}^{*}\right)\in A$ is a robust-optimization equilibrium of the robust game $\left\{ A_{i},f_{i},W_{i}^{\delta_{i}} : i \in \mathcal{N} \right\}$, then $\left(x_{1}^{*},\ldots,x_{n}^{*}\right)$ is an $\epsilon$-Nash equilibrium of its nominal counterpart, with $\epsilon=\max\left\{C^{\delta_{1}}_{1}\left(\mathbf{x}_{-1}^{*}\right),\ldots,C^{\delta_{n}}_{n}\left(\mathbf{x}_{-n}^{*}\right)\right\}$. Moreover, for $\hat{\epsilon}<\epsilon$, $\left(x_{1}^{*},\ldots,x_{n}^{*}\right)$ is not an $\hat{\epsilon}$-Nash equilibrium of the nominal counterpart game.
\end{theorem}

\medskip

Thus, in case of payoff uncertainty, uncertainty-averse players play an $\epsilon$-Nash equilibrium instead of a Nash equilibrium, where the level of approximation $\epsilon$ is the largest of players' opportunity costs of uncertainty. Moreover, reaching a better approximation is not possible in the sense that $\epsilon$ cannot be smaller than the greatest of the players' opportunity costs of uncertainty.

Measuring the level of approximation that can be obtained in a robust game, the opportunity costs of uncertainty vanish when uncertainty vanishes. Nevertheless, it is not straightforward to determine if and at which speed they shrink when uncertainty reduces. In particular, kept fixed the actions of the opponents, it cannot be stated in general that player $i$'s opportunity cost of uncertainty converges to zero when his level of uncertainty does so, i.e. when $\delta_{i}\rightarrow 0$. Moreover, given the actions of the opponents, it is not clear, for example, if player $i$'s opportunity cost of uncertainty is monotonically increasing with respect to its level of uncertainty. This because, given the actions of the opponents $\mathbf{x}_{-i}$, player $i$'s opportunity cost of uncertainty depends on $\delta_{i}$ according to a function not known in general. In fact, player $i$'s opportunity cost of uncertainty depends on the robust-optimization strategy played by player $i$ in a robust game where his/her level of uncertainty is $\delta_{i}$ and the opponents play the action profile $\mathbf{x}_{-i}$.

There is, however, a class of robust games for which the behavior of the opportunity cost of uncertainty can be characterized in more detail. In fact, under the further assumption of $f_{i}\left(\boldsymbol{\alpha}_{i};x_{i},\mathbf{x}_{-i}\right)$ concave w.r.t. $\boldsymbol{\alpha}_{i}$, $\forall \boldsymbol{\alpha}_{i}\in U_{i}$, $\forall \left(x_{i},\mathbf{x}_{-i}\right)\in A$ and $\forall i\in \mathcal{N}$, as suggested in \cite{CrespiRadiRocca2017}, it is possible to construct a function that is linear w.r.t. the parametrized level of uncertainty and that approximates by excess player $i$'s opportunity costs of uncertainty. To this aim, define
\begin{equation}\label{Efs}
E_{i}\left(x_{i},\mathbf{x}_{-i}\right):= \rho^{0}_{i}\left(x_{i},\mathbf{x}_{-i}\right)-\rho^{1}_{i}\left(x_{i},\mathbf{x}_{-i}\right)\quad \forall i \in \mathcal{N}
\end{equation}
and the function
\begin{equation}\label{Ef}
\bar{E}_{i}\left(\mathbf{x}_{-i}\right):= \max_{x_{i}\in A_{i}}E_{i}\left(x_{i},\mathbf{x}_{-i}\right) \quad \forall i \in \mathcal{N}
\end{equation}
which can be used to construct an approximation by excess of the opportunity cost of uncertainty of player $i$ as indicated in the following Lemma (see proof in \ref{Ap:T}).

\medskip

\begin{lemma}\label{Lemma::1}
Let $f_{i}\left(\boldsymbol{\alpha}_{i};x_{i},\mathbf{x}_{-i}\right)$ be concave w.r.t. $\boldsymbol{\alpha}_{i}$,  $\forall \boldsymbol{\alpha}_{i}\in U_{i}$ and  $\forall i\in \mathcal{N}$. Then, $\delta_{i}\bar{E}_{i}\left(\mathbf{x}_{-i}\right)\geq C^{\delta_{i}}_{i}\left(\mathbf{x}_{-i}\right)$ $\forall \mathbf{x}_{-i}\in A_{-i}$  and $\forall i \in \mathcal{N}$.
\end{lemma}

\medskip

By definition, an $\epsilon^{1}$-Nash equilibrium is also an $\epsilon^{2}$-Nash equilibrium as long as $\epsilon^{2}\geq \epsilon^{1}$. Hence, Lemma \ref{Lemma::1} implies the following result.\footnote{A proof of Theorem \ref{Theorem::DeltaROE} was already proposed in \cite{CrespiRadiRocca2017}. However, the one here proposed is new, straightforward and elegant and it is obtained by exploiting a property of $\epsilon$-Nash equilibria as well as Theorem \ref{Theorem::DeltaROEm}, which were not considered in \cite{CrespiRadiRocca2017}.}

\medskip

\begin{theorem}[in \cite{CrespiRadiRocca2017}]\label{Theorem::DeltaROE}
Let $f_{i}\left(\boldsymbol{\alpha}_{i};x_{i},\mathbf{x}_{-i}\right)$ be concave w.r.t. $\boldsymbol{\alpha}_{i}$,  $\forall i\in \mathcal{N}$. If $\left(x_{1}^{*},\ldots,x_{n}^{*}\right)\in A$ is a robust-optimization equilibrium of the robust game $\left\{ A_{i},f_{i},W_{i}^{\delta_{i}} : i \in \mathcal{N} \right\}$, then $\left(x_{1}^{*},\ldots,x_{n}^{*}\right)$ is an $\epsilon$-Nash equilibrium of its nominal counterpart, with $\epsilon=\max\left\{\delta_{1}\bar{E}_{1}\left(\mathbf{x}_{-1}^{*}\right),\ldots,\delta_{n}\bar{E}_{n}\left(\mathbf{x}_{-n}^{*}\right)\right\}$.
\end{theorem}

\medskip

Theorem \ref{Theorem::DeltaROE} underlines the relationship between a robust-optimization equilibrium of a robust game and an $\epsilon$-Nash equilibrium of the nominal counterpart game, where $\epsilon$ represents an upper bound of players' opportunity costs of uncertainty. Differently from the opportunity cost of uncertainty defined in \eqref{Ebm}, this upper bound is a function that is both linear and strictly increasing in the level of uncertainty. In fact, by definition of $E_{i}$ in \eqref{Efs}, the function $\bar{E}_{i}$ in \eqref{Ef} does not depend on $\delta_{i}$. Therefore, despite its being based on an upper-bound approximation of the opportunity cost of uncertainty, the result in Theorem \ref{Theorem::DeltaROE} can be useful whenever a closed form expression for the opportunity cost of uncertainty is difficult do obtain. Note that the result in Theorem \ref{Theorem::DeltaROE} can be used instead of the one in Theorem \ref{Theorem::DeltaROEm} to study the loss of profit derived by uncertainty only when the payoff functions are concave with respect to the uncertain parameters. In this respect, the scope of Theorem \ref{Theorem::DeltaROE} is very large as players' payoff functions are linear with respect to the uncertain parameters, therefore concave, in the majority of the games proposed in the literature.

A further clarification is required for a correct interpretation of Theorems \ref{Theorem::DeltaROEm} and \ref{Theorem::DeltaROE}. As already observed, the opportunity cost of uncertainty does not measure the difference between a player's profit at the Nash equilibrium and at a robust-optimization equilibrium. Then, the results in Theorems \ref{Theorem::DeltaROEm} and \ref{Theorem::DeltaROE} do not imply that players' payoff at a robust-optimization equilibrium are always lower than at a Nash equilibrium of the nominal counterpart game. Moreover, Theorems \ref{Theorem::DeltaROEm} and \ref{Theorem::DeltaROE} define a relationship between an $\epsilon$-Nash equilibrium and a robust-optimization equilibrium, but an $\epsilon$-Nash equilibrium is not necessarily a Nash equilibrium. In particular, an $\epsilon$-Nash equilibrium is also a Nash equilibrium only if that $\epsilon$-Nash equilibrium remains so when $\epsilon\rightarrow 0$, which is not in general the case. In other words, Theorems \ref{Theorem::DeltaROEm} and \ref{Theorem::DeltaROE} do not exclude that uncertainty can create virtuous configurations such that players are better off with respect to a Nash equilibrium of the nominal counterpart game. A point that will be underlined and discussed also in the next section where a robust duopoly model is considered.

\subsection{An alternative foundation for $\epsilon$-Nash equilibria}

The concept of $\epsilon$-Nash equilibrium is more general than the one of Nash equilibrium, in the sense that a Nash equilibrium is an $\epsilon$-Nash equilibrium but it is not true the vice versa. Hence, the set of $\epsilon$-Nash equilibria of a game does not necessarily include a Nash equilibrium. The same holds true for robust-optimization equilibria. In fact, Theorems \ref{Theorem::DeltaROEm} and \ref{Theorem::DeltaROE} underline that the concept of $\epsilon$-Nash equilibrium is more general than the one of robust-optimization equilibrium, i.e. the existence of a robust-optimization equilibrium of a robust game implies the existence of an $\epsilon$-Nash equilibrium of the nominal counterpart game. However, the vice versa is not true in general. Hence, a robust game does not necessarily have a robust-optimization equilibrium, which, in case it exists, it belongs to the set of $\epsilon$-Nash equilibria of the nominal counterpart game. Nevertheless, as specified in the following Theorem (see proof in \ref{Ap:T}), for each $\epsilon$-Nash equilibrium of a nominal game, it is possible to construct a robust game, by choosing suitable payoff-uncertainty sets, such that the first one is the nominal counterpart of the second one, and the $\epsilon$-Nash equilibrium of the nominal game is also a robust-optimization equilibrium of the robust game.

\medskip

\begin{theorem}\label{EspilonNashAsROE}
Consider a nominal game $\left\{ A_{i},f_{i}: i \in \mathcal{N} \right\}$. Then, for each $\epsilon$-Nash equilibrium $\left(x_{1}^{*},\ldots,x^{*}_{n}\right)$ of this game there exists at least a robust game such that $\left(x_{1}^{*},\ldots,x^{*}_{n}\right)$ is a robust-optimization equilibrium of the robust game and $\left\{ A_{i},f_{i}: i \in \mathcal{N} \right\}$ is its nominal counterpart.
\end{theorem}

\medskip

The game-theory literature motivates players' waiver of an extra $\epsilon$ profit required to have that an $\epsilon$-Nash configuration is an equilibrium, as an extra cost for searching a better solution or for changing strategy, see, e.g., \cite{Dixon1987}. The result in Theorem \ref{EspilonNashAsROE} provides a further explanation for the existence of an $\epsilon$-Nash equilibrium in terms of opportunity cost of uncertainty. Specifically, for each $\epsilon$-Nash equilibrium, the $\epsilon$ profit that players give up when they play an $\epsilon$-Nash equilibrium can be interpreted as the opportunity cost of uncertainty. As the opportunity cost of uncertainty depends on the worst-case approach to uncertainty of players, interpreting uncertainty as the lack of trust on the nominal realization, see \cite{Knight1921}, the opportunity cost of uncertainty can also be named \emph{cost of ambiguity aversion (or cost of aversion to uncertainty)}. In this case, the $\epsilon$ profit that players waive when they play an $\epsilon$-Nash equilibrium can be interpreted as the cost of aversion to uncertainty.

Regarding $\epsilon$-Nash equilibria of a nominal game, the result in Theorem \ref{EspilonNashAsROE} provides therefore a theoretical foundation for these equilibria in terms of uncertainty aversion. Moreover, it provides a selection criterion to discriminate among $\epsilon$-Nash equilibria. Indeed, defined an $\epsilon$-level of approximation, these equilibria are not generally unique. In most of the real-world applications it is the algorithm employed to solve a nominal game and its setting to discriminate the $\epsilon$-Nash equilibrium that will be considered. Having at least a theoretical justification for the arbitrary choice is a \emph{reassuring} aspect. In this respect, Theorem \ref{EspilonNashAsROE} indicates that for the $\epsilon$-approximation chosen and for the $\epsilon$-Nash equilibrium computed, there exists an uncertainty set, a level of uncertainty and a robust game such that the  $\epsilon$-Nash equilibrium can be interpreted as a robust-optimization equilibrium of the game and the $\epsilon$-approximation is the maximum of the opportunity costs of uncertainty (or the maximum of the costs of aversion to uncertainty) of the players of the robust game.

\subsection{An example of robust game and remarks on computing equilibria}

Summarizing, a robust-optimization equilibrium of a robust game is an $\epsilon$-Nash equilibrium of the nominal counterpart game, where $\epsilon$ measures the opportunity cost of uncertainty. At the same time, each $\epsilon$-Nash equilibrium of a nominal game can be interpreted as a robust-optimization equilibrium of a robust game that has the first one as its nominal counterpart. Despite these analogies, there are differences in terms of complexity and methods to adopt when computing an $\epsilon$-Nash equilibrium of a nominal game and a robust-optimization equilibrium of a robust game. In very large games, the $\epsilon$-Nash equilibria can be computed using polynomial-time algorithms, while to find a Nash solution requires non-polynomial-time algorithms in general, see, e.g., \cite{DaskalakisPapadimitriou2015}. Since searching for a robust-optimization equilibrium is equivalent to finding a Nash equilibrium once the worst-case payoff functions are derived (see again the definition of robust-optimization equilibrium above), the same computational complexity involves robust-optimization equilibria. In stylized games, on the contrary, the $\epsilon$-Nash equilibria cannot be computed with general analytical methods, while the robust-optimization equilibria can be computed employing standard techniques used for searching Nash equilibria in nominal games. For example, it is possible to construct the \emph{worst-case (or robust) best-reply functions}, which are the robust counterpart of best-reply functions and are defined as follows, see \cite{CrespiRadiRocca2017}:
\begin{equation}\label{RBRFi}
R_{i}^{\delta_{i}}\left(\mathbf{x}_{-i}\right)=\arg\max_{x_{i}\in A_{i}}\left[\min_{\boldsymbol{\alpha}_{i}\in W_{i}^{\delta_{i}}} f_{i}\left(\boldsymbol{\alpha}_{i};x_{i},\mathbf{x}_{-i}\right)\right]
\end{equation}
These functions can be used to identify the robust-optimization equilibria. In fact, the action profile $\left(x^{*}_{1},\ldots,x^{*}_{n}\right)\in A$ is a robust-optimization equilibrium if and only if 
\begin{equation}\label{ROEwcbrf}
x^{*}_{i}\in R_{i}^{\delta_{i}}\left(\mathbf{x}^{*}_{-i}\right) \quad \forall i \in \mathcal{N}
\end{equation}
If the worst-case best-reply functions can be used to compute robust-optimization equilibria of a robust game, their derivation may not be so straightforward as it is for the best-reply functions of a nominal game. To overcome the issue, a simple procedure is proposed in \ref{AlgorithmBR} to derive players' worst-case best-reply functions when the uncertainty sets are polyhedra, which is a common assumption in robust-optimization. The algorithm is used to obtain the worst-case best-reply functions in the example that follows. The example regards a simple robust game where players' payoff functions are linear-quadratic and depend on a single unknown parameter. 

\begin{example}\label{example1}
Consider the finite-person, non-cooperative, simultaneous-move, one-shot robust game $\left\{ A_{i},f_{i},W_{i}^{\delta_{i}} : i \in \left\{1,2\right\} \right\}$, where 
\begin{itemize}
\item[]
\item The action spaces are $A_{i}:=\left[0,1.8\right]\subset \mathbb{R}$, for $i=1,2$;
\item[]
\item The uncertainty sets are $W_{i}^{\delta_{i}}:=\delta_{i} U_{i}+ \left(1-\delta_{i}\right)\boldsymbol{\alpha}_{i}^{0}=\delta_{i} \left[0.1,0.8\right]+ \left(1-\delta_{i}\right)0.6\subset \mathbb{R}$, for $i=1,2$.
\item[]
\item The level of uncertainty is the same for both players, i.e. $\delta_{1}=\delta_{2}=\delta$.
\item[]
\item The payoff functions are $f_{i}\left(\boldsymbol{\alpha}_{i};x_{i},\mathbf{x}_{-i}\right):=\left(1+\boldsymbol{\alpha}_{i}\left(1-x_{i}\right)-\mathbf{x}_{-i} \right)x_{i}$, for $i=1,2$;
\item[]
\end{itemize}

This is a two-player game. The action spaces, the uncertainty sets and the payoff functions are the same for both players. Hence, the game is symmetric. For analogy with the notation in the rest of the section, variables $\boldsymbol{\alpha}_{i}$ and $\mathbf{x}_{i}$ are in bold despite being single-entry vectors. All conditions in Assumption \ref{Ass1} are satisfied when $\delta=1$, therefore at least a robust-optimization equilibrium exists for the robust game, whatever level of uncertainty $\delta$, see Theorem \ref{EquilibriumExistenceROE}. Specifically, the nominal counterpart of the game, obtained for $\delta=0$, has a unique Nash equilibrium, which is the intersection point of the best-reply functions, as observable in Figure \ref{FigI}, picture in the middle. Therefore, the nominal counterpart game is a symmetric game with a unique symmetric Nash equilibrium. This feature is not preserved in the robust version of the game when the levels of uncertainty are large. Consider for example the case $\delta = 1$ and build the worst-case best-reply functions using the Algorithm suggested in \ref{AlgorithmBR}. As observable in Figure \ref{FigI}, first picture from the left, the robust game has multiple robust-optimization equilibria, which correspond to the intersection points of the worst-case best reply functions. Moreover, only one of the seven robust-optimization equilibria is symmetric. Figure \ref{FigI}, last picture from the left, reports the projections of the equilibrium outputs of the robust game on the action space of the first player. The graphic shows that when the level of uncertainty converges to zero only one of the robust-optimization equilibria of the robust game survives and converges \emph{smoothly} to the Nash equilibrium of the nominal game. The other robust-optimization equilibria disappear. The unique robust-optimization equilibrium that survives to small level of uncertainty is the only equilibrium output of the robust game that in the next section will be classified as the robust-optimization equilibrium with a Nash equilibrium counterpart.

The example marks the large differences that could occur between a simple robust game and its nominal counterpart. These differences are inherited from the robust optimization employed by the players of the game to handle uncertainty. A robust-optimization problem can be more complicated, and can generate a different output, than its nominal correspondence. In robust optimization, the complexity of the optimization problem to solve depends on the shapes of the uncertainty sets. There are conservative formulations of the uncertainty set, where the worst-case realization of the uncertain parameters does not depend on the decision variables. In this case, the robust problem is equivalent to the nominal optimization problem by setting the values of the uncertain parameters equal to their worst-case realizations, see \cite{Soyster1973}. As the shape of the uncertainty sets becomes more complex, the complexity of the robust-optimization problem increases and the optimal solution can be substantially different from the one of the nominal optimization problem, see, e.g., \cite{Ben-TalGhaouiNemirovski2009}. This aspect of the robust optimization is reflected and magnified in robust game theory. To magnify the difference between a robust game and its nominal counterpart is the strategic interaction among players. In fact, in game theory a worst-case realization of the value of an uncertain parameter does not change as a function of the decision variables (the action) of a player only but it depends on the action of the opponents as well. This is witnessed by the current example, where the uncertainty sets considered are the simplest possible\footnote{Note that a simpler shape of an uncertainty set is only the trivial one given by a single point or realization, i.e. absence of uncertainty.}, the uncertainty regards a single parameter of the payoff functions, and despite those features the solution of the robust game for $\delta=1$ is very different from the one of the nominal counterpart obtained for $\delta=0$.

\begin{figure}
\begin{center}
\includegraphics[scale=0.45]{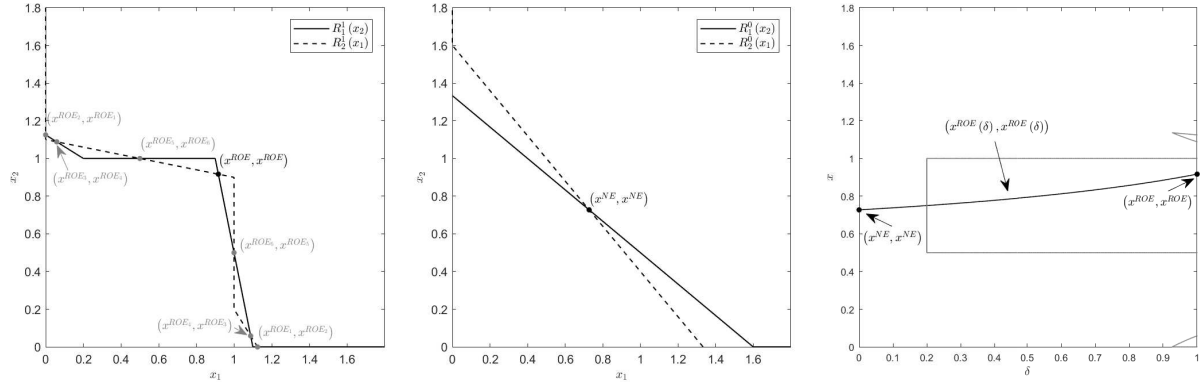}
\caption{Graphical representation of worst-case best-reply functions for the robust game in Example \ref{example1}, left panel. Graphical representation of the best-response functions of the nominal counterpart of the robust game in Example \ref{example1}, panel in the middle. Projections in the action space of player 1 of the robust-optimization equilibria of the robust game as a function of the level of uncertainty, right panel.}\label{FigI}
\end{center}
\end{figure}

\end{example}

\subsection{Measuring the effect of uncertainty}

The example underlines that not all robust-optimization equilibria converge to a Nash equilibrium. Some disappear as uncertainty reduces. Some change and become Nash equilibria when uncertainty vanishes for all players. These are equilibrium configurations that have a Nash equilibrium counterpart. In this section we consider $\delta_{i}=\delta$, for all $i\in\mathcal{N}$, for the sake of notational simplicity, and we discriminate between robust-optimization equilibria with a Nash equilibrium counterpart and robust-optimization equilibria without it.
Specifically, a Nash equilibrium is the counterpart of a robust-optimization equilibrium when reducing the level of uncertainty $\delta$ of the robust game, a continuous trajectory of robust-optimization equilibria parametrized by $\delta$ exists that converges to a Nash equilibrium. A formal definition of Nash equilibrium counterpart is provided in the following and it is essential to discuss similarities and differences between a robust game and its nominal counterpart. 

\medskip
 
\begin{definition}\label{DefinitionROENASHcounterpart}
The action profile $\mathbf{x}^{*}\left(0\right)$ is a Nash equilibrium counterpart  of a robust-optimization equilibrium $\mathbf{x}^{*}\left(\bar{\delta}\right)$ when there exists a continuous function $\chi:\left[0,\bar{\delta}\right]\rightarrow A$, such that $\chi\left(\delta\right)$ is a robust-optimization equilibrium of $\left\{ A_{i},f_{i},W_{i}^{\delta}: i \in \mathcal{N} \right\}$, $\chi\left(0\right)=\mathbf{x}^{*}\left(0\right)$ and $\chi\left(\bar{\delta}\right)=\mathbf{x}^{*}\left(\bar{\delta}\right)$.
\end{definition}

\medskip

According to the definition, the counterparty relationship occurs only between a robust-optimization equilibrium and a Nash equilibrium of the nominal counterpart game. Specifically,  given a robust game with a robust-optimization equilibrium, if a Nash equilibrium counterpart of a robust-optimization equilibrium exists, then it is a Nash equilibrium of the nominal counterpart game.

The existence of a robust-optimization equilibrium without a Nash equilibrium counterpart indicates that the equilibrium strategies of the players of a robust game can have an \emph{abrupt} change when uncertainty vanishes. This has several implications. Small changes on the level of uncertainty may imply big changes on the equilibrium output of the game. Moreover, small errors in the estimation of the level of uncertainty affecting players may cause big errors on the estimated output of the game. On the other hand, the Nash equilibrium counterpart of a robust-optimization equilibrium is not unique and a multiplicity of Nash equilibria that are counterpart of a robust-optimization equilibrium creates \emph{indeterminacy} that complicates the prediction of the effect of an uncertainty reduction on the equilibrium output of the game. This form of indeterminacy is not present in a robust game that admits a unique robust-optimization equilibrium, which has a Nash equilibrium counterpart, and the uniqueness of the equilibrium output is preserved varying the level of uncertainty. This is a robust game that does not admit surprises in the sense that small changes in the level of uncertainty will cause small changes in the equilibrium output which, being unique, is not path-dependent. In this respect, the definition of Nash-equilibrium counterpart of a robust-optimization equilibrium allows us to identify a class of robust games for which the uncertainty impacts \emph{smoothly} on the equilibrium output. More generally, for this class of games small levels of uncertainty ensure that a robust-optimization equilibrium is confined in the neighborhood of the Nash equilibrium of the nominal counterpart game.

A set of sufficient conditions to impose on the payoff functions to have a game with a unique equilibrium output, whatever level of uncertainty, may be very restrictive. Here, we avoid to identify such a set of conditions and we consider a class of robust games that have a unique robust-optimization equilibrium and its uniqueness is preserved when reducing the level of uncertainty. Then, we investigate the conditions to impose on the payoff functions for the existence of a unique Nash-equilibrium counterpart. In particular, the following theorem, the proof of which is in \ref{Ap:T}, provides sufficient conditions to identify such robust games.

\medskip

\begin{theorem}[Existence of Nash-equilibrium counterpart]\label{Th::ContinuityROE}
Consider $\bar{\delta}\in\left(0,1\right]$ and a robust game $\left\{ A_{i},f_{i},W_{i}^{\bar{\delta}} : i \in \mathcal{N} \right\}$ that satisfies the conditions in Assumption \ref{Ass1}. Assume that all the robust games $\left\{ A_{i},f_{i},W_{i}^{\delta} : i \in \mathcal{N} \right\}$, with $\delta\in\left[0,\bar{\delta}\right)$, have a unique robust-optimization equilibrium. This robust-optimization equilibrium has a Nash equilibrium counterpart when, for each $i\in \mathcal{N}$, one of the following conditions is satisfied:
\begin{itemize}
\smallskip
\item[] A1) $f_{i}\left(\boldsymbol{\alpha}_{i};x_{i},\mathbf{x}_{-i}\right)$ is strictly concave in $x_{i}\in A_{i}$, for all $\mathbf{x}_{-i}\in A_{-i}$ and $\boldsymbol{\alpha}_{i}\in W^{\bar{\delta}}_{i}$;
\smallskip
\item[] A2) $\rho^{\delta}_{i}\left(x_{i},\mathbf{x}_{-i}\right)$ is strictly concave in $x_{i}\in A_{i}$, for all $\delta\in\left[0,\bar{\delta}\right]$ and $\mathbf{x}_{-i}\in A_{-i}$.
\end{itemize}
\end{theorem}

\medskip

Imposing some further conditions on the payoff functions (in some cases of the nominal game only), the existence of a robust-optimization equilibrium and of the counterpart Nash equilibrium allows us to characterize also the behavior of the opportunity cost of uncertainty as uncertainty vanishes, which otherwise is not easy to determine as remarked above. Moreover, under these further assumptions, the existence of a Nash equilibrium counterpart is sufficient to confine a robust-optimization equilibrium of a robust game into the set of $\epsilon$-Nash equilibria of the nominal counterpart game, with $\epsilon$ arbitrary, as long as the level of uncertainty is sufficiently low as stated in the following theorem, the proof of which is in \ref{Ap:T}.

\medskip

\begin{theorem}\label{EspilonNashforCounterpartROE}
Consider a robust game with a positive level of uncertainty, i.e. $\delta>0$, which admits a robust-optimization equilibrium $\left(x^{*}_{i}\left(\delta\right),\mathbf{x}^{*}_{-i}\left(\delta\right)\right)$ which has a Nash equilibrium counterpart $\left(x^{*}_{i}\left(0\right),\mathbf{x}^{*}_{-i}\left(0\right)\right)$. Then,
\begin{itemize}
\smallskip
\item[-] The opportunity cost of uncertainty evaluated in $\mathbf{x}^{*}_{-i}\left(\delta\right)$ is a continuous function in $\delta =0$;
\item[-] For each $\epsilon>0$, there exists $\delta\left(\epsilon\right)\in\left(0,1\right)$ and there exists an $\epsilon$-Nash equilibrium of the nominal counterpart game that is also a robust-optimization equilibrium of the same robust game but with $\delta\left(\epsilon\right)$-level of uncertainty;
\end{itemize}
when one of the following conditions is satisfied:
\begin{itemize}
\smallskip
\item[] H1) for each $i\in \mathcal{N}$, $f_{i}\left(\boldsymbol{\alpha}^{0}_{i};\cdot,\cdot\right)$ is a continuous function;
\smallskip
\item[] H2)  for each $i\in \mathcal{N}$, $f_{i}\left(\boldsymbol{\alpha}_{i};x_{i},\mathbf{x}_{-i}\right)$ is concave w.r.t. $\boldsymbol{\alpha}_{i}$.
\end{itemize}
\end{theorem}

\medskip

The results in Theorem \ref{EspilonNashforCounterpartROE} identify two different sets of sufficient conditions under which the existence of a Nash equilibrium, which is counterpart of a robust-optimization equilibrium, implies that the opportunity cost of uncertainty of each player computed at the robust-optimization equilibrium converges to zero as uncertainty vanishes. These sufficient conditions regard the nominal payoff functions. Hence, studying the nominal game only, it is possible to state if the opportunity cost of uncertainty of a player vanishes as uncertainty vanishes and, for every $\epsilon$, if the set of all $\epsilon$-Nash equilibria of the nominal counterpart game contains a robust-optimization equilibrium of the robust game. In other terms, assuming valid the behavioral hypothesis that players are robust optimizers, Theorem \ref{EspilonNashforCounterpartROE} indicates that is sufficient to impose conditions on the nominal game only, to have a Nash equilibrium that is \emph{robust} with respect to small levels of uncertainty.

The results so far discussed indicate that the existence of a Nash equilibrium counterpart is relevant for several aspects. The first relevant aspect is the possibility to confine a robust-optimization equilibrium in a neighborhood of the Nash equilibrium counterpart as long as the level of uncertainty is sufficiently low and the neighborhood is identified by the set of $\epsilon$-Nash equilibria, where $\epsilon$ measures the opportunity cost of uncertainty. This indicates equilibrium configurations that are not too dissimilar between a robust game and its nominal counterpart and predictions made according to a nominal game that are \emph{robust} to small levels of uncertainty. A second related aspect is the possibility to identify a class of robust games that are less sensitive to changes in the level of uncertainty. Here sensitivity is intended as the impact of the uncertainty on the equilibrium output of the game. A third relevant aspect is the possibility to use the opportunity cost of uncertainty to proxy the loss of gain that occurs between playing a robust-optimization equilibrium and its counterpart Nash equilibrium. Indeed, under the conditions identified in Theorem \ref{EspilonNashforCounterpartROE}, reducing the level of uncertainty, each set of $\epsilon$-Nash equilibria of the nominal counterpart game contains at least a robust-optimization equilibrium of the robust game, indicating that the opportunity cost of uncertainty converges to zero as uncertainty vanishes. Therefore, as long as the level of uncertainty is sufficiently low, the opportunity cost of uncertainty is a good approximation of the loss of gain that occurs by playing a robust-optimization equilibrium instead of its counterpart Nash equilibrium.

This last point allows us to provide a complete characterization of $\epsilon$-Nash equilibria in terms of robust-optimization equilibria. On one side, the result in Theorem \ref{EspilonNashAsROE} ensures that each $\epsilon$-Nash equilibrium can be interpreted as a robust-optimization equilibrium by introducing some uncertainty on the parameters of the payoff functions; and the $\epsilon$ waiver of an extra profit can be interpreted as the opportunity cost of uncertainty.
On the other side, Theorem \ref{EspilonNashforCounterpartROE} indicates in which specification an uncertainty set represents a criterion of choice to select a subgroup of $\epsilon$-Nash equilibria. In fact, having a robust game that satisfies the conditions in Theorem \ref{EspilonNashforCounterpartROE} indicates that the set of $\epsilon$-Nash equilibria contains a robust-optimization equilibrium which can be used to forecast the $\epsilon$-Nash equilibrium that the players will play. This is a relevant point as the $\epsilon$-Nash equilibria are usually infinitely many and to forecast an $\epsilon$-Nash equilibrium that the players will play is a discretionary choice of the game designer.

Summarizing, the interpretation/justification of the $\epsilon$-approximation related to an $\epsilon$-Nash equilibrium, as the cost for searching for a better solution, does not allow to discriminate among the set of $\epsilon$-Nash equilibria. On the contrary, interpreting the $\epsilon$-Nash equilibria as robust-optimization equilibria of a robust game, we say that there exists only a subgroup of $\epsilon$-Nash equilibria that is consistent with the level of uncertainty that defines the robust game. Therefore, this theoretical foundation of $\epsilon$-Nash equilibria provides also a criterion of choice for $\epsilon$-Nash equilibria.

\subsection{Further remarks}\label{RemSum}

The study of the similarities between an equilibrium output of a robust game and the one of the nominal counterpart game leads us  to introduce the concept of counterpart Nash equilibrium. 
A counterpart Nash equilibrium exists when a robust-optimization equilibrium moves smoothly towards a Nash equilibrium of the nominal counterpart game as uncertainty vanishes. The presence of such an equilibrium indicates a form of \emph{regularity} in the way in which the game reacts to uncertainty. The existence of a robust-optimization equilibrium that converges smoothly towards a Nash equilibrium of the nominal counterpart game does not exclude, however, the existence of other equilibria that have a less regular behavior as uncertainty vanishes. In addition, the existence of a counterpart Nash equilibrium is guaranteed under stringent conditions on the payoff functions of a game. More general results (derived in this section) indicate that a robust-optimization equilibrium of a robust game can be confined in the set of $\epsilon$-Nash equilibria of the nominal counterpart game, where the $\epsilon$ approximation measures the opportunity cost of uncertainty. These theoretical results are based on a parametrization of the uncertainty set and allow to measure the sensitivity of a game with respect to the level of uncertainty. This is only one feature of the uncertainty that influences the behavior of a robust player. The second element is the shape of the uncertainty set. The representation of the uncertainty proposed in this work is indeed characterized by these two elements. In forecasting or policy analysis, it is also relevant to measure the sensitivity of the equilibrium outputs of a robust game with respect to the shape of the uncertainty set. In fact, a large sensitivity to the configuration of the uncertainty set by the equilibrium output of a robust game can lead to wrong forecasts when the uncertainty is misspecified. Hence, the study of robust game theory here proposed is far from being complete.

The discussion underlines that the validity of the results of this paper is therefore confined to the assumption of a correct specification of the shape of the uncertainty set. In fact, only under this hypothesis, the theoretical findings on the analogies between robust games and their nominal counterparts allow us to measure the effect of uncertainty. However, independently of the correct specification of the shape of uncertainty set, the results in this paper do not allow us to neglect uncertainty. The risk is of inaccurate forecasts and misleading deductions. Consider, for example, a nominal game that admits only symmetric equilibrium configurations and a robust version of it that admits asymmetric robust-optimization equilibria. This qualitative discordance may lead to think of a wrong characterization of the nominal game. It occurs every time we observe robust players play an asymmetric robust-optimization equilibrium while the nominal game can predict only a symmetric equilibrium output. It is therefore relevant the possibility that a robust-optimization equilibrium of a robust game can be a Nash equilibrium of a nominal game obtained by considering a realization included in the uncertainty sets of the parameters with unknown value. If this is not the case, then the robust game can generate equilibrium outputs that are difficult even to imagine, not only to forecast, by studying the nominal counterpart game. Emblematic is Example \ref{example1} in this section, which underlines that robust-optimization equilibria characterized by asymmetric actions of the players can occur, but preserving the symmetry property of the robust game, such equilibrium-output configurations are not possible in the nominal game. Moreover, even though a single robust-optimization equilibrium coincides with a Nash equilibrium of a nominal version of the game, the entire set of robust-optimization equilibria may not be replicated by a single nominal version of the game. In fact, in a robust game the worst-case realizations of the parameters depend on the actions undertaken by the players, and this feature of robust games may make it not possible that all the robust-optimization equilibria of a robust game are also Nash equilibria of a nominal game obtained by considering a single constellation for the values of the parameters of the robust game.

Summarizing, the narrative in this section goes in the direction to underline analogies, similarities and differences between a robust game and its nominal counterpart. In the following we consider an economic application of the robust game theory. Specifically, a robust version of a Cournot duopoly model is considered. The duopoly model has a setup similar to the one proposed in \cite{SinghVives1984}. The investigation of the model aims to underline and discuss the complexity that emerges in robust games even considering very simple settings. Consistently with the scope of the example, only a simple configuration of the robust version of the Cournot duopoly model is considered.

\section{An application: A Cournot duopoly game with payoff uncertainty and robust firms}\label{Application}

In this section we consider a Cournot duopoly game with linear inverse demand functions, non constant marginal costs of production and differentiated products. The production is totally sold in the market which is characterized by a representative consumer that expresses a certain degree of substitutability between the two products. The two firms (robust players) that populate the duopoly produce differentiated goods, with each firm that produces one type of output only as in \cite{SinghVives1984}, and are uncertain about the slopes of the inverse demand functions as well as the degree of substitutability of the two products. Denote by $q_{1}$ the production of firm 1 and by $q_{2}$ the production of firm 2. The feasible levels of production for firm $i$ are represented by $A_{i}$, which is a subset of $\mathbb{R}$.

The price at which firm $i$ sells its output (or commodity) $q_{i}$, with $i=1,2$, is given by $P\left(q_{i},q_{-i}\right)=\max\left\{\beta_{i}-s_{i}q_{i}-\gamma q_{-i};0\right\}$, where $q_{-i}$ is the production of the competitor, $\beta_{i}\in \mathbb{R}_{+}$ is the choke price of the commodity produced by firm $i$, $s_{i}\in \mathbb{R}_{+}$ is the price sensitivity of output of firm $i$ with respect to its own production, and $\gamma\in \left[0,s_{i}\right]$ is the degree of substitutability between the two products. The cost of production of output $q_{i}$ is given by $C_{i}\left(q_{i}\right)=c_{i}q_{i}+d_{i}q^2_{i}$, where $c_{i}\in \mathbb{R}_{+}$ and $d_{i}\in \mathbb{R}$.\footnote{For firm $i$, $d_{i}=0$ means constant marginal costs of production (constant return to scale technology), $d_{i}>0$ means increasing marginal costs of production (decreasing return to scale technology), and $d_{i}<0$ means decreasing marginal costs of production (increasing return to scale technology).} Confining the firms' action spaces to levels of production for which the price functions are positive and setting $b_{i}=s_{i}+d_{i}$ and $a_{i}=\beta_{i}-c_{i}$, firm $i$'s profit (or payoff) function is then given by:
\begin{equation}\label{Profitfunctions}
f_{i}\left(\boldsymbol{\alpha}_{i};q_{i},q_{-i}\right)= \left(a_{i} - \gamma  q_{-i} - b_{i}q_{i}\right)q_{i}
\end{equation}
where $\boldsymbol{\alpha}_{i}=\left(b_{i},\gamma\right)$ is the vector of the uncertain parameters that affects the profit (payoff) function of each firm (player) $i$.

Firms (or players) know that the inverse demand functions are linear and downward sloping. However, the reactivity of the market price to changes in the level of production, as well as the degree of substitutability of the products, are uncertain. In particular, adopting the parametric representation of uncertainty proposed above, the uncertainty set of player $i$ is given by $W_{i}^{\delta_{i}}=\delta_{i} U_{i}+\left( 1-\delta_{i} \right) \boldsymbol{\alpha}_{i}^{0}$, where $\delta_{i} \in \left[ 0,1\right]$ measures the level of uncertainty, $\boldsymbol{\alpha}_{i}^{0}$ is the singleton $( \widehat{b}_{i},\widehat{\gamma })$, while $U_{i}$ is a non-empty compact set that includes all the possible values of $\boldsymbol{\alpha}_{i}=\left(b_{i},\gamma \right)$, i.e. it represents the maximum level of uncertainty. Hence, $\boldsymbol{\alpha}_{i}^{0}\subset U_{i}$.

The shape of the uncertainty set $U_{i}$ should reflect the economic relationship among the variables of the game, for example some variables are positively related, others negatively related. However, the uncertainty set does not need to be known by a robust player, who, being a worst-case maximizer, needs to know only the subset made by the worst-case parameter realizations, see Section \ref{ExistenceROE}. This subset does not need to reflect the shape of the uncertainty set, hence the choice of a specific worst-case frontier is only apparently an arbitrary assumption. For example, assuming a worst-case frontier represented by a downward sloping segment in the $b_{i}-\gamma$ plane, it  does not imply that lower values of $b_{i}$ are associated to higher values of $\gamma$ and vice versa. See Figure \ref{Fig1} for an example of three different uncertainty sets that, despite having the same worst-case frontier, indicate different economic relationships between the parameters $b_{i}$ and $\gamma$.

\begin{figure}
\begin{center}
\includegraphics[scale=0.6]{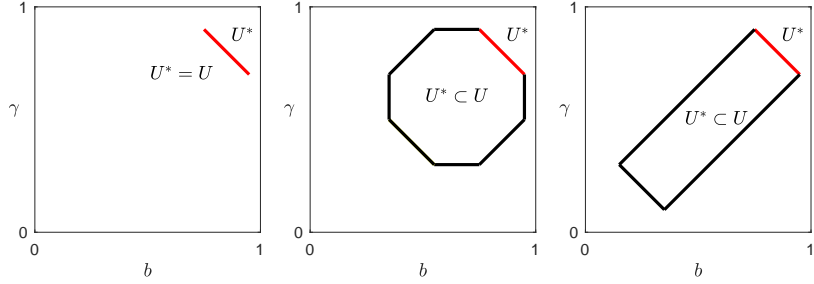}
\caption{Examples of uncertainty sets with the same worst-case frontier. First picture from the left, example of uncertainty set equal to the worst-case frontier. Second picture from the left, uncertainty set of octagonal shape. Third picture from the left, uncertainty set of rectangular shape. In the pictures, the uncertainty set is denoted by $U$ and its worst-case frontier by $U^{*}$.}\label{Fig1}
\end{center}
\end{figure}

In the light of these considerations, in the following only the worst-case realizations will be defined and the economic interpretation of the worst-case frontier is neglected as it is not relevant. In particular, $U_{i}\subset \mathbb{R}^{2}$ is assumed to be made of the segment joining the two worst-case realizations $\left(\bar{b}_{i},\underline{\gamma}\right)$ and $\left(\underline{b}_{i},\bar{\gamma}\right)$, with $\bar{b}_{i}>\underline{b}_{i}$ and $\bar{\gamma}>\underline{\gamma}$. This is one of the simplest representations of a worst-case frontier that makes the outputs of the robust game different from the output of the nominal game. A simpler worst-case frontier is made by a single worst-case realization, which is the case when the uncertainty set is component-wise and each parameter with unknown value impacts either positively or negatively the payoff function independently of the actions of the players. As discussed above, this implies that the robust game is the nominal game with each parameter of the payoff functions that is set at the unique worst-case realization.

Underlying the differences that can occur between a robust game and its nominal counterpart is one of the aims of this example. Consistently with this aim, only the symmetric version of the duopoly is considered. Indeed, this setting allows to better underline analogies and differences between a robust game and its nominal counterpart. Hence, the following restrictions are imposed.

\medskip

\begin{assumption}\label{Assumption1}
These restrictions hold true in the following:
\begin{enumerate}
\item The nominal game is symmetric: Parameters are such that $a_{1}=a_{2}=a$, $\widehat{b}_{1}=\widehat{b}_{2}=\widehat{b}$ and the action spaces are such that $A_{1}=A_{2}=\tilde{A}$;
\item The robust game is symmetric: $U_{1}=U_{2}=U$ (which implies $\underline{b}_{1}=\underline{b}_{2}=\underline{b}$ and $\bar{b}_{1}=\bar{b}_{2}=\bar{b}$) and $\delta_{1}=\delta_{2}=\delta$;
\item All the parameters are non-negative: $a,\bar{b},\underline{b},\bar{\gamma},\underline{\gamma}\geq 0$;
\item The action space $\tilde{A}$ is a non-empty, closed, bounded and convex subset of $\mathbb{R}_{\geq0}$;
\item $2\widehat{b}>\widehat{\gamma}$ (technical condition that implies a unique Nash equilibrium for the nominal game);
\item The set $U$ is defined as follows:
\begin{equation}\label{Uibar}
U=\left\{ \boldsymbol{\alpha}=\left(b,\gamma\right) \ | \ \bar{b} \geq b \geq\underline{b} 
\ \wedge \ \gamma =\frac{\bar{\gamma}\bar{b}-\underline{b}\underline{\gamma}}{\bar{b}-\underline{b}} - \frac{\bar{\gamma}-\underline{\gamma}}{\bar{b} -\underline{b}}b\right\}
\end{equation}
\end{enumerate}
\end{assumption}

\medskip

According to Assumption \ref{Assumption1}, $W^{\delta_{1}}_{1}=W^{\delta_{2}}_{2}=W^{\delta}$, with $W^{\delta}$ which is given by:
\begin{equation}\footnotesize
W^{\delta}=\left\{ \left( b,\gamma \right) \  | \ \bar{b}\left(\delta\right)\geq b\geq\underline{b} \left(\delta\right)
\ \wedge \ \gamma =\frac{\bar{\gamma}\left(\delta\right)\bar{b}\left(\delta\right)-\underline{b}\left(\delta\right)\underline{\gamma}\left(\delta\right)}{\bar{b}\left(\delta\right)-\underline{b}\left(\delta\right)} -\frac{\bar{\gamma}\left(\delta\right)-\underline{\gamma}\left(\delta\right)}{\bar{b}\left(\delta\right)-\underline{b}\left(\delta\right)}b\left(\delta\right)   \right\} 
\end{equation}
where
\begin{eqnarray}
\overline{b}\left( \delta\right)  &=&\left( 1-\delta \right) \widehat{b}+ \delta\overline{b}\\
\underline{b}\left( \delta \right)  &=&\left( 1-\delta \right)\widehat{b}+\delta \underline{b} \\
\overline{\gamma }\left( \delta \right)  &=&\left( 1-\delta \right) \widehat{\gamma }
+\delta \overline{\gamma } \\
\underline{\gamma }\left( \delta \right)  &=&\left( 1-\delta \right) \widehat{\gamma }
+\delta \underline{\gamma }
\end{eqnarray}
Hence $W^{\delta}$ is a non-empty, closed, bounded and convex subset of $\mathbb{R}^{2}$.

Consistently with the definition of robust game proposed in Section \ref{RobustGames}, the \emph{robust Cournot duopoly game} can therefore be defined as follows:
\begin{equation}\label{G}
\left\{A_{i},f_{i},W_{i}^{\delta_{i}}:i\in\left\{1,2\right\}\right\}
\end{equation}
where the dependence of $W^{\delta}$, $f$ and $A$ on $i$ can be dropped because of the assumption of symmetric game. By Definition \ref{WCWPF}, the \emph{worst-case payoff function} of firm (player) $i$ is given by
\begin{equation}
\rho^{\delta}\left(q_{i},q_{-i}\right)=\min_{\boldsymbol{\alpha}\in W ^{\delta}}f \left(\boldsymbol{\alpha};q_{i},q_{-i}\right) 
\end{equation}
The restrictions imposed in Assumption \ref{Assumption1} ensure that action spaces of firms as well as the uncertainty sets respect the conditions imposed in Assumption \ref{Ass1}. In addition, under the parameter value restrictions imposed in Assumption \ref{Assumption1}, it is straightforward to verify that the payoff function of firm $i$ defined in \eqref{Profitfunctions} is continuous and it is concave with respect to $q_{i}$ for all $\left(\boldsymbol{\alpha},q_{-i}\right)\in U\times \tilde{A}$. Therefore, independently of the level of uncertainty, the worst-case payoff function of player $i$ is well-defined, continuous and concave with respect to the action space of the player himself, and by Theorem \ref{EquilibriumExistenceROE} a robust-optimization equilibrium of the robust duopoly game exists as well as a Nash equilibrium of the nominal counterpart game.

To verify the number of robust-optimization equilibria that exist and to compute them, we derive the worst-case best reply functions. For $\delta>0$, the worst-case best reply (or reaction) function for robust firm $i$ is defined as follows:\footnote{The analytical expression of the worst-case best reply function can be obtained by employing the algorithm in \ref{AlgorithmBR}.}
\begin{equation}\label{RBR}
R^{\delta}\left(q_{-i}\right): =\arg \max_{q_{i}\in A }\rho^{\delta}\left(q_{i},q_{-i}\right)=
\left\{ 
\begin{array}{lcr}
\frac{a}{2\overline{b}\left( \delta \right) }-\frac{\underline{\gamma }\left( \delta \right) }{2\overline{b}\left( \delta\right) }q_{-i} &\quad  \text{if} \quad \quad  & \underline{q}\left( \delta \right)> \text{ }q_{-i}\geq0 \\
\\
\frac{\overline{\gamma }  -\underline{\gamma } }{\overline{b}  -\underline{b} }q_{-i} & \quad  \text{if} \quad\quad  & \overline{q}\left( \delta \right)\geq
q_{-i}\geq\underline{q}\left( \delta \right) \\ 
\\
\frac{a}{2\underline{b}\left( \delta\right) }-\frac{\overline{\gamma }\left( \delta\right) }{2\underline{b}\left( \delta\right) }q_{-i}\text{ } & \quad  \text{if} \quad\quad  & q^{M}\left( \delta \right)> q_{-i}>\overline{q}\left( \delta \right) \\ 
\\
0 & \quad  \text{if} \quad\quad  & q_{-i}\geq q^{M}\left( \delta \right)%
\end{array}
\right. 
\end{equation}
where
\begin{eqnarray} \label{qunderlinei}
\underline{q}\left( \delta \right) &=&\frac{a\left( \overline{b}  -\underline{b} 
\right) }{2\overline{b}\left( \delta\right) \left( \overline{\gamma }  -\underline{\gamma } 
\right)+\underline{\gamma }\left( \delta\right) \left( \overline{b}  -\underline{b} \right)} \\ \label{qbari}
\overline{q}\left( \delta \right) &=&\frac{a \left( \overline{b}  -\underline{b} 
\right) }{2\underline{b}\left( \delta\right) \left( \overline{\gamma }  -\underline{\gamma }  \right) +\overline{\gamma }\left( \delta\right) \left( 
\overline{b} -\underline{b} \right) } \\ \label{qMi}
q^{M}\left( \delta \right) &=&\frac{a}{\bar{\gamma}\left( \delta\right) }
\end{eqnarray}

In accordance to the definition of robust-optimization equilibrium in \eqref{ROEwcbrf}, see also \cite{AghassiBertsimas2006} and \cite{CrespiRadiRocca2017}, given an uncertainty level $\delta\in\left[0,1\right]$ and the robust best-reply function in \eqref{RBR}, the output set $\left(q^{*}_{i},q^{*}_{-i}\right)\in A$ is a robust-optimization equilibrium of the robust game \eqref{G} if and only if
\begin{equation}\label{defCournotROE}
q_{i}^{*}= R^{\delta}\left( q^{*}_{-i}\right) \quad \forall i \in \left\{1,2\right\} 
\end{equation}
where ``$=$'' substitutes ``$\in$'' in \eqref{ROEwcbrf}, as for robust game here considered the robust best replies are functions instead of correspondences. 
Hereafter, a robust-optimization equilibrium of the robust duopoly model \eqref{RBR} will be called \emph{Cournot robust-optimization equilibrium} (Cournot ROE in short), which becomes a Cournot-Nash equilibrium when uncertainty vanishes.

Setting $\delta=0$, uncertainty vanishes and the robust duopoly game with Cournot competition becomes a nominal Cournot duopoly game where players' behavior is characterized by a classical best-reply function:
\begin{equation}\label{BRnom}
R^{0}\left( q_{-i}\right) =\left\{ 
\begin{array}{lcr}
\frac{a}{2\widehat{b} }-\frac{\widehat{
\gamma }}{2\widehat{b} }q_{-i} & \quad  \text{if} \quad\quad  & q^{M}\left(0\right)> q_{-i}>0 \\ 
\\
0 & \quad  \text{if} \quad\quad  & q_{-i}\geq q^{M}\left(0\right)
\end{array}%
\right. 
\end{equation}

Solving system \eqref{defCournotROE} when $\delta=0$, it results that the unique Cournot-Nash equilibrium of the nominal counterpart of the duopoly game is the one provided in the following proposition (see proof in \ref{AppendixB}).

\medskip

\begin{proposition}[Cournot-Nash equilibrium of the nominal Cournot duopoly]\label{NENCDG}
Consider Assumption \ref{Assumption1}. The nominal version of the Cournot duopoly game, i.e. game \eqref{G} with $\delta=0$, admits one and only one Cournot-Nash equilibrium which is given by
\begin{equation}\label{NE:NominalGame}
\left(q^{NE},q^{NE}\right)=\left(\frac{a}{2\widehat{b}+\widehat{\gamma}},\frac{a}{2\widehat{b}+\widehat{\gamma }}\right)
\end{equation}
\end{proposition}

\medskip

The strategic profile in \eqref{NE:NominalGame} represents a symmetric Nash equilibrium; all the players involved play the same Nash strategy, which is unique. Symmetry and uniqueness of the equilibrium solution are not guaranteed in the robust duopoly game. A striking feature of the robust duopoly model arises comparing the robust reaction function \eqref{RBR} with its nominal counterpart \eqref{BRnom}. The latter one is monotonically decreasing while the robust counterpart is in general a non-monotone function. This difference reflects in the number of the robust-optimization equilibria of the robust game, which can be multiple and even asymmetric as stated in the following proposition (see proof in \ref{AppendixB}).

\medskip

\begin{proposition}[Cournot-robust-optimization equilibria of the robust Cournot duopoly]\label{ROERCDG}
Consider the robust Cournot duopoly in \eqref{G}, with $\delta>0$. If the shape of the uncertainty set $U$ is such that:
\begin{enumerate}
\item $\overline{b}-\underline{b}>\overline{\gamma }-\underline{\gamma}$, then
\begin{equation}
\left(q^{ROE_{1}},q^{ROE_{1}}\right)=\left(\frac{a}{2\overline{b}\left( \delta \right) +\underline{\gamma}\left( \delta \right)} , \frac{a}{2
\overline{b}\left( \delta\right) + \underline{\gamma }\left( \delta\right) }\right)
\end{equation}
is the unique robust-optimization equilibrium of the robust Cournot duopoly in \eqref{G} and it converges to Cournot-Nash equilibrium of its nominal counterpart game when $\delta\rightarrow 0$.
\item $\overline{b}-\underline{b}=\overline{\gamma}-\underline{\gamma}$, then robust-optimization equilibria fill the interval $E_{1}=\left\{ \left(q,q\right) |\overline{q}\geq q\geq \underline{q}\right\}$. For $\delta\rightarrow 0$, the segment $E$ shrinks into the Cournot-Nash equilibrium $\left(q^{NE},q^{NE}\right)$ in \eqref{NE:NominalGame}.
\item $\overline{b}-\underline{b}<\overline{\gamma}-\underline{\gamma}$ and the level of uncertainty is such that
\begin{itemize}
\item[i)] $\delta<\delta^{*}$, where
\begin{equation}
\delta^{*} = \frac{\overline{\gamma}-2\underline{b}}{\overline{\gamma}-2\underline{b}+2\widehat{b}-\widehat{\gamma}}\left(<1\right),
\end{equation}
then
\begin{equation}
\left(q^{ROE_{2}},q^{ROE_{2}}\right)=\left(\frac{a}{2
\underline{b}\left( \delta\right) +\overline{\gamma }\left( \delta\right) }, \frac{a}{2
\underline{b}\left( \delta\right)+  \overline{\gamma }
\left( \delta\right) }
\right)
\end{equation}
is the unique robust-optimization equilibrium of the Cournot duopoly in \eqref{G} and it converges to Cournot-Nash equilibrium of its nominal counterpart when $\delta\rightarrow 0$.
\item[ii)] $\delta=\delta^{*}$, then robust-optimization equilibria fill the segment $E_{2}=\left\{ \left(q,\frac{a}{2\underline{b}\left( \delta\right) }-q\right) | q^{M}>\frac{a}{2\underline{b}\left( \delta\right) }-q,q>\overline{q}\right\}$.
\item[iii)] $\delta^{*}<\delta<1$ ($\delta^{*}>0$ is equivalent to $2\underline{b} < \overline{\gamma }$), then
$\left(q^{ROE_{2}},q^{ROE_{2}}\right)$,
\begin{equation}
\footnotesize
\left(q^{ROE_{3}},q^{ROE_{4}}\right)= \left(
\frac{a \left( \overline{b}  -\underline{b}  \right) }{2\underline{b}\left( \delta\right) \left( \overline{b}  -\underline{b}  \right) +\overline{\gamma }\left( \delta\right) \left( \overline{\gamma } -\underline{\gamma }  \right) }
,
\frac{a \left( \overline{\gamma }  -\underline{\gamma }  \right) }{2%
\underline{b}\left( \delta\right) \left( \overline{b}  -\underline{b}  \right) +%
\overline{\gamma }\left( \delta\right) \left( \overline{\gamma }  -\underline{\gamma } \right) }%
\right)
\end{equation}
 and $\left(q^{ROE_{4}},q^{ROE_{3}}\right)$ are robust-optimization equilibrium of the Cournot duopoly in \eqref{G}.
\end{itemize}
\end{enumerate}
\end{proposition}

\medskip

Proposition \ref{ROERCDG} underlines that a robust Cournot duopoly can have multiple robust-optimization equilibria, see, e.g., Figure \ref{FigureBR2}. Therefore, in a robust Cournot duopoly game there may even be uncertainty about the Cournot-robust-optimization equilibrium the firms play. The multiplicity of equilibria vanishes as uncertainty vanishes: The nominal counterpart of the game admits a unique Cournot-Nash equilibrium, as indicated in Proposition \ref{NENCDG}. In addition, Proposition \ref{ROERCDG} underlines that the coexistence of Cournot-robust-optimization equilibria requires a certain level of uncertainty as well as a certain shape of the uncertainty set.

The existence of multiple robust-optimization equilibria is only one peculiarity of the robust duopoly model \eqref{G}. The second peculiarity is the presence of asymmetric robust-optimization equilibria (where a robust firm produces more/less than the other one) despite the assumption of identical robust players. The existence of asymmetric equilibria in a symmetric game is not prerogative of robust games. It occurs in nominal games as well. After all, a robust game can be seen as a nominal game as shown in the previous sections. The interesting point to underline here is indeed another one, that is the difference between a robust game that admits asymmetric equilibria and its nominal counterpart game that does not. This aspect has relevant economic implications. In fact, predictions made according to a nominal game can be misleading and uncertainty can have a strong impact on the final outcome of a game.

Summarizing, as the current duopoly model shows, multiple equilibria instead of a unique output and asymmetric equilibria instead of symmetric outputs can be obtained by introducing uncertainty in a game. It is also worth underlining that the worst-case frontier considered is among the simplest ones. Thus, the differences between the robust duopoly model and its nominal counterpart is not the result of an ad-hoc worst-case frontier of the uncertainty set. A more complicated shape of the worst-case frontier could imply an even more marked contrast between the output of robust game and the one of its nominal counterpart, suggesting that using a nominal game to infer the output of a robust game can lead to wrong forecasts.

The presence of multiple and asymmetric robust-optimization equilibria in a symmetric robust duopoly game implies heterogeneous levels of production for the two identical firms. This implies that a firm experiences a maximum guaranteed payoff (profit) at the equilibrium output which is higher than the one of the competitor, as specified in the following proposition (see proof in \ref{AppendixB}). 

\medskip

\begin{proposition}\label{AsROEProfits}
Assume that $\left(q^{ROE_{3}},q^{ROE_{4}}\right)$ is a robust-optimization equilibrium for game \eqref{G}. Playing this robust-optimization equilibrium, robust firm (or player) $1$ produces less then the competitor and it records a lower maximum-guaranteed profit. The opposite holds true in $\left(q^{ROE_{4}},q^{ROE_{3}}\right)$.
\end{proposition}

\medskip

It comes out that in a robust Cournot duopoly game, which is symmetric and whose nominal counterpart has a unique symmetric Cournot-Nash equilibrium, an increase of uncertainty may generate asymmetry in the final outcome of the game and the earnings are not equally distributed. In other words, the aversion to uncertainty can cause two identical firms to coordinate to play an asymmetric robust-optimization equilibrium where one firm produces more and records a higher maximum guaranteed payoff. Moreover, the asymmetric robust-optimization equilibrium does not have a nominal counterpart, consistently with Definition \ref{DefinitionROENASHcounterpart}. This confirms that uncertainty and ambiguity aversion can favor forms of coordination between firms non consistent with the ones predicted by a nominal game. On the contrary, the symmetric Cournot-robust-optimization equilibrium has a Cournot-Nash equilibrium counterpart. As indicated in Proposition \ref{ROERCDG}, the Cournot-robust-optimization equilibrium $\left(q^{ROE_{1}},q^{ROE_{1}}\right)$ converges \emph{smoothly} towards the Cournot-Nash equilibrium $\left(q^{NE},q^{NE}\right)$ when the level of uncertainty goes to zero. For a low level of uncertainty, the existence of a Cournot-robust-optimization equilibrium with a Cournot-Nash equilibrium counterpart is ensured by Theorem \ref{Th::ContinuityROE}. In fact the robust-duopoly model here considered satisfies all the conditions imposed in Theorem \ref{Th::ContinuityROE}.

Concerning the payoffs of the players at the robust-optimization equilibria, an asymmetric robust-optimization equilibrium ensures to a firm a higher profit than the one the same firm would experience by playing the symmetric robust-optimization equilibrium.  The same asymmetric robust-optimization equilibrium is not convenient for the competitor. In particular, we have the following result (see proof in \ref{AppendixB}). 

\medskip

\begin{proposition}\label{AsROEProfitsC}
Assume that $\left(q^{ROE_{3}},q^{ROE_{4}}\right)$, $\left(q ^{ROE_{2}},q^{ROE_{2}}\right)$ and $\left(q^{ROE_{4}},q^{ROE_{3}}\right)$ are robust-optimization equilibria, then
\begin{equation}
\rho^{\delta}\left(q^{ROE_{4}},q^{ROE_{3}}\right) \geq \rho^{\delta}\left(q^{ROE_{2}},q ^{ROE_{2}}\right) \geq \rho^{\delta}\left(q^{ROE_{3}},q^{ROE_{4}}\right) 
\end{equation}
\end{proposition}

\medskip

Therefore, for firm $2$ it is convenient (in terms of maximum guaranteed payoff) to play the asymmetric robust-optimization equilibrium $\left(q^{ROE_{3}},q^{ROE_{4}}\right)$ rather than the symmetric robust-optimization equilibrium $\left(q^{ROE_{2}},q^{ROE_{2}}\right)$. On the contrary, for firm $1$ it is convenient to play the symmetric robust-optimization equilibrium $\left(q^{ROE_{2}},q^{ROE_{2}}\right)$, rather than the asymmetric robust-optimization equilibrium $\left(q^{ROE_{3}},q^{ROE_{4}}\right)$. Symmetric arguments suggest that the vice versa is true when the asymmetric robust-optimization equilibrium $\left(q^{ROE_{4}},q^{ROE_{3}}\right)$ is considered in place of $\left(q^{ROE_{3}},q^{ROE_{4}}\right)$. 

As a further remark on this point, in Figure \ref{FigureBR2}, third picture from the left, we can observe that the nominal profit of a firm at a robust-optimization equilibrium is larger than the one at the Nash equilibrium. Despite the profit of a firm at the Nash equilibrium being lower than the one at a robust-optimization equilibrium, it is always convenient to reduce one's own level of uncertainty. This benefit can also be quantified. As underlined in the previous section, each firm (robust player) associates an opportunity cost of uncertainty to each possible action profile of this robust duopoly game. This opportunity cost of uncertainty measures the maximum amount that a single firm would pay to eliminate uncertainty when the competitor does not modify its own level of production. According to definition \eqref{Ebm}, the opportunity cost of uncertainty for firm $i$ is given by
\begin{equation}
C^{\delta}\left(q_{-i}\right)=\rho^{0}\left(R^{0}\left(q_{-i}\right),q_{-i}\right)-\rho^{0}\left(R^{\delta}\left(q_{-i}\right),q_{-i}\right)
\end{equation}
where the best-reply function $R^{\delta}$ is as in \eqref{RBR}.

The behavior of firm $i$'s opportunity cost of uncertainty at a robust-optimization equilibrium, as a function of the level of uncertainty, is drawn in Figure \ref{FigureBR2}, picture in the middle. Considering the symmetric Cournot-robust-optimization equilibrium, the opportunity cost of uncertainty converges to zero when the level of uncertainty of the game vanishes. This indicates that the Cournot-robust-optimization equilibrium converges \emph{smoothly} to the Cournot-Nash equilibrium of the nominal game when uncertainty vanishes, i.e. the Cournot-Nash equilibrium is its nominal counterpart. The asymmetric Cournot-robust-optimization equilibria disappear instead when uncertainty shrinks, see Proposition \ref{ROERCDG}, therefore they do not have a Cournot-Nash equilibrium counterpart. Nevertheless, the opportunity cost of uncertainty at these equilibria shrinks when uncertainty reduces. This is not a surprise, as for the robust Cournot duopoly game the results in Theorem \ref{Theorem::DeltaROE} are valid since the payoff functions are concave with respect to the uncertain parameters.

As a final remark, the theoretical results developed in the previous section, specifically Theorem \ref{Theorem::DeltaROEm}, indicate that in equilibrium robust players play an $\epsilon$-Nash equilibrium, where $\epsilon$ measures the maximum of players' opportunity cost of uncertainty. To underline this aspect, consider a parameter constellation such that three Cournot-robust-optimization equilibria exist. Figure \ref{FigureBR2}, first picture from the left, shows that the Cournot-robust-optimization equilibria of the duopoly game are $\epsilon$-Cournot-Nash equilibria.


\begin{figure}
\begin{center}
\includegraphics[scale=0.45]{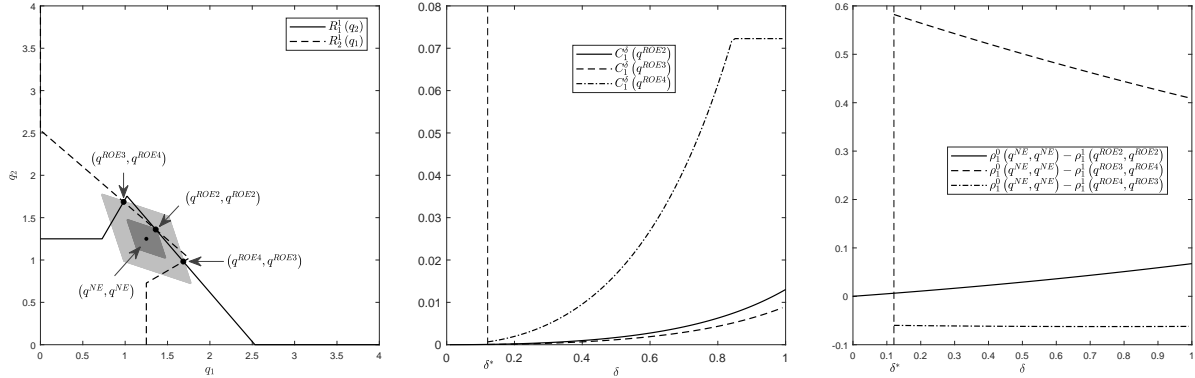}
\caption{First picture form the left, we observe in light-gray the set of all $\epsilon$-Nash equilibria of the nominal duopoly game, with $\epsilon = \max\left\{C_{1}^{\delta}\left(q^{ROE_{4}}\right),C_{2}^{\delta}\left(q^{ROE_{3}}\right)\right\}$, which contains the robust-optimization equilibrium $\left(q^{ROE_{3}},q^{ROE_{4}}\right)$, see Theorem \ref{Theorem::DeltaROEm}, and $\left(q^{ROE_{4}},q^{ROE_{3}}\right)$ by symmetry. In dark gray the set of all $\epsilon$-Nash equilibria of the nominal duopoly game, with $\epsilon = \max\left\{C_{1}^{\delta}\left(q^{ROE_{2}}\right),C_{2}^{\delta}\left(q^{ROE_{2}}\right)\right\}$, which contains the robust-optimization equilibrium $\left(q^{ROE_{2}},q^{ROE_{2}}\right)$, see Theorem \ref{Theorem::DeltaROEm}. Second picture from the left, the opportunity cost of uncertainty for firm $i$ as a function of the level of uncertainty $\delta$. The third picture form the left, the different of profit that a firm obtain at a robust-optimization equilibrium with respect to the Nash equilibrium of the game.}\label{FigureBR2}
\end{center}
\end{figure}

At the end of this section, it is worth noting the modeling flexibility offered by the robust game theory. As already pointed out, the worst-case best-reply functions are not monotonically decreasing despite the assumptions of a linear inverse demand function and of linear marginal costs. This is a relevant point. In fact non-monotone reaction functions are assumed by many, see for example \cite{Rand1978}, to show the inherent chaotic nature of a Cournot system. However, microeconomic foundations of Cournot duopoly games characterized by non-monotonic reaction curves usually require ad-hoc assumptions, such as marginal costs represented by cubic functions see, e.g. \cite{Furth1986}, or cost functions with an inter-firm externalities, see, e.g., \cite{PostonStewart1978}. This opens the puzzle on 
which market structures are capable of generating non-descending best reply functions. Robust game theory shows that non-monotonic reaction curves are possible even assuming simple market structures as long as there is some uncertainty about payoffs and firms are worst-case maximizers.

\section{Summary and conclusions}\label{Conc}

A robust game is a generalization of a normal form game where the worst-case approach is employed by players to handle uncertainty about payoff environment and this uncertainty is represented by a bounded set of possible realizations. This modeling framework is complementary to the ambiguity-game framework where players have only vague beliefs about opponents' actions and define strategies that maximize their Choquet expected utilities. The theoretical results in this work underline the analogies, similarities and differences between a robust game and its nominal counterpart.

The first aim has been to provide theoretical foundations for robust games by investigating the link with nominal games. The main results are summarized in the following. Given a robust game $\left\{ A_{i},f_{i},W_{i}^{\delta_{i}} : i \in N \right\}$, we have that
\begin{itemize}
\smallskip
\item $\left\{ A_{i},f_{i},W_{i}^{0} : i \in N \right\}$ is a nominal game, it can be rewritten as $\left\{ A_{i},f_{i} : i \in N \right\}$, and it is called nominal counterpart game;
\smallskip
\item $\left(x_{1}^{*},\ldots,x_{n}^{*}\right)$ is a robust-optimization equilibrium of the robust game if and only if it is a Nash equilibrium of the (nominal) game $\left\{ A_{i},\rho_{i}^{\delta_{i}} : i \in N \right\}$, where $\rho_{i}^{\delta_{i}}$ are the worst-case payoff functions defined in \eqref{WCWPF};
\smallskip
\item Given the action of the opponents $\mathbf{x}_{-i}$, player $i$'s opportunity cost of uncertainty is measured by function $C^{\delta_{i}}_{i}\left(\mathbf{x}_{-i}\right)$ defined in \eqref{Ebm};
\smallskip
\item A robust-optimization equilibrium of a robust game is an $\epsilon$-Nash equilibrium of the nominal counterpart game, where $\epsilon$ is the maximum of players' opportunity cost of uncertainty at the robust-optimization equilibrium, see Theorem \ref{Theorem::DeltaROEm}.
\smallskip
\item If $\left(x_{1}^{*},\ldots,x_{n}^{*}\right)$ is a robust-optimization equilibrium of a robust game, $\epsilon$ defined as in Theorem \ref{Theorem::DeltaROEm} measures the uncertainty aversion at the equilibrium in terms of opportunity cost of uncertainty;
\smallskip
\smallskip
\item Each $\epsilon$-Nash equilibrium of a nominal game can be seen as a robust-optimization equilibrium of a robust game where the first is its nominal counterpart, see Theorem \ref{EspilonNashAsROE}.
\smallskip
\item Under some regularity assumptions, player $i$'s opportunity cost of uncertainty converges to zero as the uncertainty of the game vanishes and a robust-optimization equilibrium of the robust game converges to a Nash equilibrium of the nominal counterpart game.
\smallskip
\end{itemize}

The second aim of this paper has been to show that neglecting players' aversion towards the uncertainty of the environment of the game may be dangerous. Indeed, the possible differences in terms of equilibrium outputs between a robust game and its nominal counterpart may be many.

A natural extension of the current work is to develop a general game-theoretical framework that accounts for both ambiguity on the opponents' actions and uncertainty of the payoff environment of the game. An intermediate research project could be limited to introducing soft-robust games where players' uncertain payoff environment is described by a distribution that belongs to some set of probability measures and players take actions merging robust and stochastic approach as in \cite{BenTalBertsimasBrown2010}. The robust-game framework can be generalized to account for different degrees of uncertainty aversion. As a player may be pessimistic looking for the maximum guaranteed payoff, another one may be optimistic looking for the maximum profit possible. Experimental tests of the robust-game setup is another relevant research topic. Important is also to study and to model the time evolution of the uncertainty aversion when historical data on opponents' actions and performances become available.   

\medskip

\begin{acknowledgements}
The authors are grateful for helpful comments and suggestions by Herbert Dawid, Michael Kopel, Luca Lambertini, Arsen Palestini, Jan Wenzelburger and by participants of the following
Workshops, Conferences and Research Seminars: the INFORMS Annual Meeting in Houston~(2017); the 14th Viennese Conference on Optimal Control and Dynamic Games, Vienna~(2019); XLII AMASES Annual Meeting in Naples~(2018); the 10th Workshop Dynamic Models in Economics and Finance -- MDEF in Urbino~(2018); the 19th Annual SAET Conference in Ischia~(2019). The research was supported by V\v{S}B-TU Ostrava under the SGS Project SP2018/34. All remaining mistakes are our own. 
\end{acknowledgements}

\medskip

\textbf{Compliance with Ethical Standards:}
The authors declare that they have no conflict of interest.

\appendix

\renewcommand\thesection{Appendix \Alph{section}}

\addcontentsline{toc}{section}{Appendices}

\section{- Proofs of results on robust games.}\label{Ap:T}

\begin{proof}[Proof of Theorem \ref{EquilibriumExistenceROEandNASHequivalentGame}] This proof of the Theorem follows the one proposed in \cite{AghassiBertsimas2006} and is reported here only for the sake of completeness. We will proceed by constructing a point-to-set mapping that satisfies the conditions of Kakutani's Fixed Point Theorem, see \cite{Kakutani1941}, and whose fixed points are precisely the equilibria of the robust game. To begin, clearly, $A=\times_{i\in\mathcal{N}}A_{i}$ is non-empty, closed, bounded, and convex, since $A_{i}$ is, $\forall i \in\mathcal{N}$. Define $\Phi:A\rightarrow 2^{A}$ as
\begin{equation}
\Phi\left(x_{1},\ldots,x_{n}\right)=\left\{\left(y_{1},\ldots,y_{n}\right)\in A| y_{i}\in\arg\max_{u_{i}\in A_{i}}\rho^{\delta_{i}}_{i}\left(u_{i},\mathbf{x}_{-i}\right), i=1,\ldots,n\right\}
\end{equation}
Let us show that $\Phi\left(x_{1},\ldots,x_{N}\right)\neq\emptyset$, $\forall\left(x_{1},\ldots,x_{N}\right)\in A$. By hypothesis of $\rho^{\delta_{i}}_{i}$ continuous in $A$ for all $i\in\mathcal{N}$, it follows that $\forall i$, $\forall \mathbf{x}_{-i}\in A_{-i}$ fixed, $\rho^{\delta_{i}}_{i}\left(x_{i},\mathbf{x}_{-i}\right)$ is continuous on $A_{i}$, a non-empty, closed, and bounded subset of an Euclidean space. Thus, by Weierstrass' Extreme Value Theorem,
\begin{equation}
\arg \max_{u_{i}\in A_{i}}\rho^{\delta_{i}}_{i}\left(u_{i},\mathbf{x}_{-i}\right)\neq\emptyset
\end{equation}
Accordingly, $\forall \left(x_{1},\ldots,x_{n}\right)\in A$,
\begin{equation}
\Phi\left(x_{1},\ldots,x_{n}\right)\neq\emptyset
\end{equation}
It is obvious from the definition of $\Phi$, that $\forall\left(x_{1},\ldots,x_{n}\right)\in A$, $\Phi\left(x_{1},\ldots,x_{n}\right)\subseteq A$, and that $\left(x_{1},\ldots,x_{n}\right)$ is an equilibrium of the robust game if and only if it is a fixed point of $\Phi$. Thus, we need only prove the existence of a fixed point of $\Phi$. Let us therefore establish that $\Phi$ satisfies the remaining conditions of Kakutani's Fixed Point Theorem; that is, we must show that $\Phi$ maps $A$ into a family of closed, convex sets, and that $\Phi$ is upper semi-continuous.
Let us first prove that, $\forall\left(x_{1},\ldots,x_{n}\right)\in A$, $\Phi\left(x_{1},\ldots,x_{n}\right)$ is a convex set. Suppose
\begin{equation}
\left(u_{1},\ldots,u_{n}\right),\left(v_{1},\ldots,v_{n}\right)\in\Phi\left(x_{1},\ldots,x_{N}\right)
\end{equation}
Then, by the definition of $\Phi$, $\forall i\in\mathcal{N}$, $\forall y_{i}\in A_{i}$,
\begin{equation}
\rho^{\delta_{i}}_{i}\left(u_{i},\mathbf{x}_{-i}\right)=\rho^{\delta_{i}}_{i}\left(v_{i},\mathbf{x}_{-i}\right)\geq\rho^{\delta_{i}}_{i}\left(y_{i},\mathbf{x}_{-i}\right)
\end{equation}
It follows that, $\forall \lambda \in \left[0,1\right]$, $\forall y_{i}\in A_{i}$,
\begin{equation}
\lambda\rho^{\delta_{i}}_{i}\left(u_{i},\mathbf{x}_{-i}\right)+\left(1-\lambda\right)\rho^{\delta_{i}}_{i}\left(v_{i},\mathbf{x}_{-i}\right)\geq \rho^{\delta_{i}}_{i}\left(y_{i},\mathbf{x}_{-i}\right)
\end{equation}
By the hypothesis of $\rho^{\delta_{i}}_{i}$ concave on $A_{i}$, it follows that
\begin{equation}
\lambda \left(u_{1},\ldots,u_{n}\right)+\left(1-\lambda\right)\left(v_{1},\ldots,v_{n}\right)\in\Phi\left(x_{1},\ldots,x_{n}\right)
\end{equation}
Hence $\Phi\left(x_{1},\ldots,x_{n}\right)$ is a non-empty and convex subset of $A$.
Let us now show that $\Phi$ is upper semi-continuous. Suppose that, for $k=1,2,3,\ldots,$
\begin{equation}
\left(x_{1}^{k},\ldots,x_{n}^{k}\right)\in A
\end{equation}
\begin{equation}
\left(y_{1}^{k},\ldots,y_{n}^{k}\right)\in \Phi\left(x_{1}^{k},\ldots,x_{n}^{k}\right)
\end{equation}
\begin{equation}
\lim_{k\rightarrow\infty}\left(x_{1}^{k},\ldots,x_{n}^{k}\right)= \left(u_{1},\ldots,u_{n}\right)\in A
\end{equation}
\begin{equation}
\lim_{k\rightarrow\infty}\left(y_{1}^{k},\ldots,y_{n}^{k}\right)= \left(v_{1},\ldots,v_{n}\right)\in A
\end{equation}
By the definition of $\Phi$, we know that, $\forall k=1,2,3,\ldots$, $\forall i \in \mathcal{N}$ and $\forall w_{i}\in A_{i}$,
\begin{equation}
\rho^{\delta_{i}}_{i}\left(y_{i}^{k},\mathbf{x}_{-i}^{k}\right)\geq\rho^{\delta_{i}}_{i}\left(w_{i},\mathbf{x}_{-i}^{k}\right)
\end{equation}
Taking the limit of both sides
\begin{equation}
\lim_{k\rightarrow +\infty}\rho^{\delta_{i}}_{i}\left(y_{i}^{k},\mathbf{x}_{-i}^{k}\right)\geq \lim_{k\rightarrow +\infty}\rho^{\delta_{i}}_{i}\left(w_{i},\mathbf{x}_{-i}^{k}\right)
\end{equation}
by hypothesis of continuity of $\rho^{\delta_{i}}_{i}$ on $A$, we obtain that, $\forall i \in\mathcal{N}$ and $\forall w_{i}\in A_{i}$,
\begin{equation}
\rho^{\delta_{i}}_{i}\left(v_{i},\mathbf{u}_{-i}\right)\geq\rho^{\delta_{i}}_{i}\left(w_{i},\mathbf{u}_{-i}\right)
\end{equation}
Hence,
\begin{equation}
\left(v_{1},\ldots,v_{n}\right)\in \Phi\left(u_{1},\ldots,u_{n}\right)
\end{equation}
and $\Phi$ is upper semi-continuous. Note that the fact that $\Phi\left(x_{1},\ldots,x_{N}\right)$ is closed follows from the fact that $\Phi$ is upper semi-continuous. This completes the proof that $\Phi$ satisfies the conditions of Kakutani's Fixed Point Theorem, and thereby establishes the existence of an equilibrium in the robust game.
\end{proof}

\begin{proof}[Proof of Lemma \ref{PropWorstCasePayoffFunctions}] Since $f_{i}$ are continuous on $W^{\delta_{i}}_{i}\times A$ by hypothesis, $W^{\delta_{i}}_{i}\times A$ is bounded by hypothesis, and the worst-case payoff function is defined as
\begin{equation}
\rho^{\delta_{i}}_{i}\left(x_{i},\mathbf{x}_{-i}\right)= \min_{\boldsymbol{\alpha}_{i}\in W^{\delta_{i}}_{i}}f_{i}\left(\boldsymbol{\alpha}_{i};x_{i},\mathbf{x}_{-i}\right)
\end{equation}
the continuity of $\rho^{\delta_{i}}_{i}$ follows from Berge's Maximum Theorem, see \cite[pp. 115-117]{Berge1963}. Consider $x_{i}\in A_{i}$ and $y_{i} \in A_{i}$. By the hypothesis of the convexity of $A_{i}$ and the concavity of $f_{i}$ on $A_{i}$, for each $\lambda\in\left[0,1\right]$ and for each $\mathbf{x}_{-i}\in A_{-i}$ we have:
\begin{equation}
\begin{array}{lll}
\rho^{\delta_{i}}_{i}\left(\lambda x_{i}+\left(1-\lambda\right)y_{i},\mathbf{x}_{-i}\right) & = & \min_{\boldsymbol{\alpha}_{i}\in W^{\delta_{i}}_{i}}f\left(\boldsymbol{\alpha}_{i};\lambda x_{i}+\left(1-\lambda\right)y_{i},\mathbf{x}_{-i}\right)\\
\\
& = & f\left(\boldsymbol{\alpha}^{*}_{i};\lambda x_{i}+\left(1-\lambda\right)y_{i},\mathbf{x}_{-i}\right)\\
\\
& \geq & \lambda f\left(\boldsymbol{\alpha}^{*}_{i}; x_{i},\mathbf{x}_{-i}\right) + \left(1-\lambda\right)f\left(\boldsymbol{\alpha}^{*}_{i}; y_{i},\mathbf{x}_{-i}\right)\\
\\
& \geq & \lambda \min_{\boldsymbol{\alpha}_{i}\in W^{\delta_{i}}_{i}}f\left(\boldsymbol{\alpha}_{i}; x_{i},\mathbf{x}_{-i}\right) + \left(1-\lambda\right)\min_{\boldsymbol{\alpha}_{i}\in W^{\delta_{i}}_{i}}f\left(\boldsymbol{\alpha}_{i}; y_{i},\mathbf{x}_{-i}\right)\\
\\
& \geq &\lambda\rho^{\delta_{i}}_{i}\left( x_{i},\mathbf{x}_{-i}\right)+\left(1-\lambda\right)\rho^{\delta_{i}}_{i}\left( y_{i},\mathbf{x}_{-i}\right)
\end{array}
\end{equation}
which proves the concavity of $\rho^{\delta_{i}}_{i}\left(\cdot,\mathbf{x}_{-i}\right)$ for every $\mathbf{x}_{-i}\in A_{-i}$.
\end{proof}

\begin{proof}[Proof of Theorem \ref{EquilibriumExistenceROE}]
The conditions in Assumption \ref{Ass1} imply that the worst-case payoff functions satisfy the conditions imposed in Theorem \ref{EquilibriumExistenceROEandNASHequivalentGame}, see Lemma \ref{PropWorstCasePayoffFunctions}. Then, the first point of the theorem follows (the first point is also considered in \cite{CrespiRadiRocca2017}, the proof is reported here for the sake of completeness). 
Concerning the second point. If $\left\{ A_{i},f_{i},W_{i}^{\delta_{i}} : i \in \mathcal{N} \right\}$ satisfies the conditions in Assumption \ref{Ass1}, then $\left\{ A_{i},f_{i},W_{i}^{\bar{\delta}_{i}} : i \in \mathcal{N} \right\}$, with $\bar{\delta}_{i}\in\left[0,\delta_{i}\right]$, satisfies the conditions in Assumption \ref{Ass1} since $W^{\bar{\delta}_{i}}\subseteq W^{\bar{\delta}_{i}}$. Hence, the proof of the second point of theorem follows form the first one.
\end{proof}

\begin{proof}[Proof of Theorem \ref{Theorem::DeltaROEm}]
Given a generic strategic profile $\left(x_{1}^{+},\ldots,x_{n}^{+}\right)\in A$, we have that
\begin{equation}
\rho^{0}_{i}\left(x^{+}_{i},\mathbf{x}^{+}_{-i}\right) \geq  \rho^{0}_{i}\left(x_{i},\mathbf{x}^{+}_{-i}\right) - \max_{x_{i}\in A_{i}}\rho^{0}_{i}\left(x_{i},\mathbf{x}^{+}_{-i}\right)+\rho^{0}_{i}\left(x^{+}_{i},\mathbf{x}^{+}_{-i}\right)  \text{, } \quad \forall x_{i} \in A_{i} \ \text{ and } \ \forall i\in \mathcal{N}
\end{equation}
Then, let $ \left(x_{1}^{*},\ldots,x_{n}^{*}\right) \in A$ be a robust-optimization equilibrium of $\left\{ A_{i},f_{i},W_{i}^{\delta_{i}} : i \in \mathcal{N} \right\}$, it follows that
\begin{equation}
\rho^{0}_{i}\left(x^{*}_{i},\mathbf{x}^{*}_{-i}\right) \geq  \rho^{0}_{i}\left(x_{i},\mathbf{x}^{*}_{-i}\right) - \max_{x_{i}\in A_{i}}\rho^{0}_{i}\left(x_{i},\mathbf{x}^{*}_{-i}\right)+\rho^{0}_{i}\left(x^{*}_{i},\mathbf{x}^{*}_{-i}\right)= \rho^{0}_{i}\left(x_{i},\mathbf{x}^{*}_{-i}\right) -C^{\delta_{i}}_{i}\left(\mathbf{x}^{*}_{-i}\right)  \text{, } \quad \forall i\in \mathcal{N}
\end{equation}
where $C_{i}^{\delta_{i}}\left(\mathbf{x}_{-i}\right)$ is defined in \eqref{Ebm}. Set $\epsilon =\max\left\{C^{\delta_{1}}_{1}\left(\mathbf{x}_{-1}^{*}\right),\ldots, C^{\delta_{n}}_{n}\left(\mathbf{x}_{-n}^{*}\right)\right\}$, we get
\begin{equation}\label{eqproof}
\rho^{0}_{i}\left(x_{i}^{*},\mathbf{x}_{-i}^{*}\right)\geq\rho^{0}_{i}\left(x_{i},\mathbf{x}_{-i}^{*}\right)-\epsilon  \text{, } \quad \forall x_{i} \in A_{i} \ \text{ and } \ \forall i\in \mathcal{N}
\end{equation}
Since $\rho^{0}_{i}$ corresponds to the payoff function of agent $i$ in the nominal counterpart game, for every $i \in \mathcal{N}$, \eqref{eqproof} proves the first result of the theorem. To prove that $ \left(x_{1}^{*},\ldots,x_{n}^{*}\right) \in A$ is not an $\hat{\epsilon}$-Nash equilibrium, with $\hat{\epsilon}\in\left(0,\epsilon\right)$, note that by definition $C_{i}^{\delta_{i}}\left(\mathbf{x}^{*}_{-i}\right)\geq 0$. If $C_{i}^{\delta_{i}}\left(\mathbf{x}^{*}_{-i}\right)=0$ for all $i\in\mathcal{N}$, then $\epsilon=0$ by definition and $\hat{\epsilon}\in\left(0,\epsilon\right)$ cannot exist. If $\exists$ $i\in\mathcal{N}$ such that $C_{i}^{\delta_{i}}\left(\mathbf{x}^{*}_{-i}\right)>0$, then $\exists$ $\left(x^{o}_{i},\mathbf{x}^{*}_{-i}\right)\in A$, and $\left(x^{o}_{i},\mathbf{x}^{*}_{-i}\right)\neq \left(x^{*}_{i},\mathbf{x}^{*}_{-i}\right)$, such that $\rho_{i}^{0}\left(x^{o}_{i},\mathbf{x}^{*}_{-i}\right)- \rho_{i}^{0}\left(x^{*}_{i},\mathbf{x}^{*}_{-i}\right)=\epsilon >0$. It follows that $\left(x^{*}_{i},\mathbf{x}^{*}_{-i}\right)$ cannot be an $\hat{\epsilon}$-Nash equilibrium, with $\hat{\epsilon}\in\left(0,\epsilon\right)$. This proves the second result of the theorem. 
\end{proof}

\begin{proof}[Proof of Lemma \ref{Lemma::1}]
From the assumption of concavity of $f_{i}\left(\boldsymbol{\alpha}_{i};x_{i},\mathbf{x}_{-i}\right)$ w.r.t. $\boldsymbol{\alpha}_{i}$,  $\forall i\in \mathcal{N}$, we have that
\begin{equation}
\rho_{i}^{0}\left(x_{i},\mathbf{x}_{-i}\right)-\rho_{i}^{\delta_{i}}\left(x_{i},\mathbf{x}_{-i}\right)-\delta_{i}\left( \rho_{i}^{0}\left(x_{i},\mathbf{x}_{-i}\right)-\rho_{i}^{1}\left(x_{i},\mathbf{x}_{-i}\right)  \right)\leq 0
\end{equation}
Then, by the sub-additivity property of the maximum operator, it follows that 
\begin{equation}
\begin{array}{lll}
\max_{x_{i}\in A_{i}}\rho_{i}^{0}\left(x_{i},\mathbf{x}_{-i}\right) & \leq &  \max_{x_{i}\in A_{i}}\rho_{i}^{\delta_{i}}\left(x_{i},\mathbf{x}_{-i}\right)+\max_{x_{i}\in A_{i}} \delta_{i} \left( \rho_{i}^{0}\left(x_{i},\mathbf{x}_{-i}\right)-\rho_{i}^{1}\left(x_{i},\mathbf{x}_{-i}\right)  \right)\\
\\
&\leq &\rho_{i}^{0}\left(x^{+}_{i}\left(\delta_{i}\right),\mathbf{x}_{-i}\right)+\max_{x_{i}\in A_{i}}\delta_{i}\left( \rho_{i}^{0}\left(x_{i},\mathbf{x}_{-i} \right)-\rho_{i}^{1}\left(x_{i},\mathbf{x}_{-i}\right)  \right)
\end{array}
\end{equation}
where $x^{+}_{i}\left(\delta_{i}\right)$ is as in \eqref{Ebm} and the last inequality follows from the property of the worst-case payoff function (see Property \ref{Prop1}). This proves the lemma.
\end{proof}

\begin{proof}[Proof of Theorem \ref{EspilonNashAsROE}]
By assumption $\left(x_{1}^{*},\ldots,x^{*}_{n}\right)$ is an $\epsilon$-Nash equilibrium of the nominal game $\left\{ A_{i},f_{i}: i \in \mathcal{N} \right\}$. Then, define $\left\{ A_{i},\bar{f}_{i},W_{i}^{\delta_{i}} : i \in \mathcal{N} \right\}$ such that:
\begin{equation}
\bar{f}_{i}\left(\boldsymbol{\alpha}_{i};x_{i},\mathbf{x}_{-i}\right)=f_{i}\left(x_{i},\mathbf{x}_{-i}\right)-\boldsymbol{\alpha}_{i}\mathbf{1}_{\left\{x_{i}\neq x^{*}_{i}\right\}}\text{,}\quad \forall i \in \mathcal{N}
\end{equation}
where $\mathbf{1}_{\left\{\cdot\right\}}$ is the indicator function, and $W_{i}^{\delta_{i}}=\delta_{i} \left[0,H\right]$, with $H>\epsilon$.
Set $\delta_{i} =\frac{\epsilon}{H}$, the worst-case payoff functions $\rho^{\delta_{i}}_{i}$ of the robust game $\left\{ A_{i},\bar{f}_{i},W_{i}^{\delta_{i}} : i \in \mathcal{N} \right\}$ are such that $\rho^{\delta_{i}}_{i}\left(x^{*}_{i},\mathbf{x}^{*}_{-i}\right)=f_{i}\left(x^{*}_{i},\mathbf{x}^{*}_{-i}\right)$ and $\rho^{\delta_{i}}_{i}\left(x_{i},\mathbf{x}_{-i}\right)=f_{i}\left(x_{i},\mathbf{x}_{-i}\right)-\epsilon$ otherwise. It follows that $\left(x^{*}_{i},\mathbf{x}^{*}_{-i}\right)$ is a robust-optimization equilibrium of the robust game $\left\{ A_{i},\bar{f}_{i},W_{i}^{\delta_{i}} : i \in \mathcal{N} \right\}$, which reduces to the nominal game $\left\{ A_{i},f_{i}: i \in \mathcal{N} \right\}$ when $\delta_{i} =0$. This completes the proof.
\end{proof}

\bigskip

\begin{proof}[Proof of Theorem \ref{Th::ContinuityROE}]
Note that condition A1) implies condition A2), see Lemma \ref{PropWorstCasePayoffFunctions}. Therefore, it is sufficient to prove that the statement of the theorem holds true under condition A2). Moreover, considering the case $\bar{\delta}=1$, the generalization to the case $\bar{\delta}\in\left(0,1\right]$ is straightforward. Choose $\delta\in\left[0,1\right]$ arbitrarily. By the conditions imposed in Assumption \ref{Ass1} and since by hypothesis $\rho^{\delta}_{i}\left(x_{i},\mathbf{x}_{-i}\right)$ are strictly concave w.r.t. $x_{i}$ for all $\mathbf{x}_{-i}\in A_{-i}$, the (robust) best reply $R^{\delta}_{i}$ defined in \eqref{RBRFi}, with $i=1,\ldots,n$, are single valued, i.e. they are real valued functions instead of more general correspondences: $R^{\delta}_{i}:A_{-i}\rightarrow A_{i}$. In fact, it is well-known that the maximum set of a strictly concave function is either empty or single-valued. Since the conditions imposed in Assumption \ref{Ass1} ensure that it is not empty, then it is single valued. Moreover, due to the assumption of continuity of $\rho_{i}^{\delta}$ w.r.t $\left(x_{i},\mathbf{x}_{-i}\right)$, by the Berge's Maximum Theorem we known that $R^{\delta}_{i}\left(\mathbf{x}_{-i}\right)$, for all $i=1,\ldots,n$, are upper hemicontinuous in $\mathbf{x}_{-i}\in A_{-i}$. Since it is well-known that a single-valued correspondence is upper hemicontinuous if and only if it is continuous as a function, we have that $R^{\delta}_{i}\left(\mathbf{x}_{-i}\right)$, for all $i=1,\ldots,n$, are continuous functions in $\mathbf{x}_{-i}\in A_{-i}$. Consider now $R^{\delta}_{i}\left(\mathbf{x}_{-i}\right)$ as a function of $\left(\mathbf{x}_{-i},\delta\right)$, i.e. $R^{\delta}_{i}:A_{-i}\times\left[0,1\right]\rightarrow A_{i}$. By the same argument and noting that $\delta$ defines a linear and convex combination of $\boldsymbol{\alpha}_{i}\in W_{i}^{\bar{\delta}}$, the further hypothesis of $f_{i}\left(\boldsymbol{\alpha}_{i};x_{i},\mathbf{x}_{-i}\right)$ continuous on $W_{i}^{\bar{\delta}}\times A$ ensures $R^{\delta}_{i}\left(\mathbf{x}_{-i}\right)$, for all $i=1,\ldots,n$, continuous in $\left(\mathbf{x}_{-i},\delta\right)$. Define $T\left(\mathbf{x},\delta\right)=\left(R^{\delta}_{1}\left(\mathbf{x}_{-1}\right),\ldots,R^{\delta}_{n}\left(\mathbf{x}_{-n}\right) \right)$. By construction $T\left(\mathbf{x},\delta\right): A\times\left[0,1\right]\rightarrow A$ is a single valued correspondence (a function), it is continuous, and the robust-optimization equilibria of the robust game $\left\{ A_{i},f_{i},W_{i}^{\delta} : i \in N \right\}$ are the solution of $T\left(\mathbf{x},\delta\right)=\mathbf{x}$. Hence, by Theorem \ref{EquilibriumExistenceROE} it follows that for each $\delta\in\left[0,1\right]$ there is a solution (not necessarily unique), labeled by $\mathbf{x}^{*}\left(\delta\right)$, of $T\left(\mathbf{x},\delta\right)=\mathbf{x}$. Moreover, 
by hypothesis that the robust game has a unique robust-optimization equilibrium, there exists a unique solution of $T\left(\mathbf{x},\delta\right)=\mathbf{x}$, labeled by $\mathbf{x}^{*}\left(\delta\right)$. Consider a $\delta^{*}\in\left[0,1\right]$. 
Suppose $\left\{\delta^{k}\right\}$ is an arbitrary sequence converging to $\delta^{*}$, as $k\rightarrow +\infty$, and consider the set $B:=\left\{\mathbf{x}^{*}\left(\delta^{k}\right):k=1,2,\ldots\right\}\subset A$, where $\mathbf{x}^{*}\left(\delta^{k}\right)$ is solution of $T\left(\mathbf{x},\delta^{k}\right)=\mathbf{x}$ in $A$ (i.e. it is a robust-optimization equilibrium of $\left\{ A_{i},f_{i},W_{i}^{\delta^{k}} : i \in \mathcal{N} \right\}$). Since $A$ is complete, there exists a subsequence $\left\{v^{k}\right\}$ of $\left\{\delta^{k}\right\}$ such that $\left\{\mathbf{x}^{*}\left(v^{k}\right)\right\}$ is convergent to some $\mathbf{z}\in A$. By the assumption of continuity of $T$ in $\left(\mathbf{x},\delta\right)$ for each $\mathbf{x}\in A$, it follows that $\mathbf{z} =T\left(\mathbf{z},\delta\right)$. Since $\mathbf{x} =T\left(\mathbf{x},\delta\right)$ has a unique solution by the hypothesis of a unique robust-optimization equilibrium for the robust game, we have $\mathbf{z}=\mathbf{x}^{*}\left(\delta^{*}\right)$. This implies that $\mathbf{x}^{*}\left(\delta^{k}\right)$ converges to $\mathbf{x}^{*}\left(\delta^{*}\right)$. Since $\left\{\delta^{k}\right\}$ is an arbitrary sequence converging to $\delta^{*}$, we have that for each $\delta \in \left[0,1\right]$, equation $T\left(\mathbf{x},\delta\right)=\mathbf{x}$ has a solution {(at least a solution)} $\mathbf{x}^{*}\left(\delta\right)\in A$ and $\mathbf{x}^{*}\left(\delta\right)\rightarrow \mathbf{x}^{*}\left(\delta^{*}\right)$ as $\delta\rightarrow \delta^{*}$. By the arbitrariness of $\delta^{*}\in\left[0,1\right]$ and since $\mathbf{x}^{*}\left(\delta\right)$ is a robust-optimization equilibrium for the robust game with $\delta$-level of uncertainty and $\mathbf{x}^{*}\left(0\right)$ is the Nash equilibrium of the nominal counterpart game by definition of $T$, the proof of the theorem is complete.
\end{proof}

\bigskip
 
\begin{proof}[Proof of Theorem \ref{EspilonNashforCounterpartROE}]
Assume H1) holds. By hypothesis, a robust-optimization equilibrium exists, let us denote it by $\mathbf{x}^{*}\left(\delta\right)$ which has a Nash-equilibrium counterpart $\mathbf{x}^{*}\left(0\right)$, which implies $\mathbf{x}^{*}\left(\delta\right)\rightarrow \mathbf{x}^{*}\left(0\right)$, as $\delta\rightarrow 0^{+}$. By definition, see \eqref{Ebm}, the opportunity cost of uncertainty of player $i$ is given by
\begin{equation}
C_{i}^{\delta}\left(\mathbf{x}^{*}_{-i}\left(\delta\right)\right)=\max_{x_{i}\in A_{i}}\rho_{i}^{0}\left(x_{i},\mathbf{x}_{-i}^{*}\left(\delta\right)\right)- \rho_{i}^{0}\left(x_{i}^{*}\left(\delta\right),\mathbf{x}_{-i}^{*}\left(\delta\right)\right)
\end{equation}
It holds:
\begin{equation}
\begin{array}{lll}
\lim_{\delta \rightarrow 0^{+}}C_{i}^{\delta}\left(\mathbf{x}^{*}_{-i}\left(\delta\right)\right) & =& \lim_{\delta\rightarrow 0^{+}}\max_{x_{i}\in A_{i}}\rho_{i}^{0}\left(x_{i},\mathbf{x}_{-i}^{*}\left(\delta\right)\right)-\lim_{\delta\rightarrow 0^{+}}\rho_{i}^{0}\left(x^{*}_{i}\left(\delta\right),\mathbf{x}_{-i}^{*}\left(\delta\right)\right)\\
\\
& = & \lim_{\delta\rightarrow 0^{+}}\max_{x_{i}\in A_{i}}\rho_{i}^{0}\left(x_{i},\mathbf{x}_{-i}^{*}\left(\delta\right)\right)-\rho_{i}^{0}\left(x^{*}_{i}\left(0\right),\mathbf{x}_{-i}^{*}\left(0\right)\right)
\end{array}
\end{equation}
The last inequality follows from the continuity assumption on $f_{i}\left(\boldsymbol{\alpha}_{i}^{0};\cdot,\cdot\right):=\rho_{i}^{0}\left(\cdot,\cdot\right)$, while the existence of the maximum follows from the continuity assumption on $\rho_{i}^{0}\left(\cdot,\cdot\right)$ and compactness of $A_{i}$. Take a sequence $\delta_{n}\rightarrow 0^{+}$ as $n\rightarrow +\infty$ and consider
\begin{equation}
x_{i}\left(\delta_{n}\right) \in \arg\max_{x_{i}\in A_{i}}\rho_{i}^{0}\left(x_{i},\mathbf{x}^{*}_{-i}\left(\delta_{n}\right)\right)
\end{equation}
Since $A_{i}$ is compact, without loss of generality we can assume $x_{i}\left(\delta_{n}\right) \rightarrow \bar{x}_{i}\in A_{i}$. The continuity assumption on $\rho^{0}_{i}$ gives
\begin{equation}
\lim_{n\rightarrow +\infty}\rho_{i}^{0}\left(x_{i}\left(\delta_{n}\right),\mathbf{x}^{*}_{-i}\left(\delta_{n}\right)\right)\rightarrow \rho_{i}^{0}\left(\bar{x}_{i},\mathbf{x}^{*}_{-i}\left(0\right)\right)
\end{equation}
Then, from 
\begin{equation}
\rho_{i}^{0}\left(x_{i}\left(\delta_{n}\right),\mathbf{x}^{*}_{-i}\left(\delta_{n}\right)\right)\geq \rho_{i}^{0}\left(x_{i},\mathbf{x}^{*}_{-i}\left(\delta_{n}\right)\right)\text{, } \quad \forall x_{i}\in A_{i}
\end{equation}
it follows
\begin{equation}
\rho_{i}^{0}\left(\bar{x}_{i},\mathbf{x}^{*}_{-i}\left(0\right)\right)\geq \rho_{i}^{0}\left(x_{i},\mathbf{x}^{*}_{-i}\left(0\right)\right) \text{, }\quad \forall x_{i}\in A_{i}
\end{equation}
which gives
\begin{equation}
\bar{x}_{i} \in \arg\max_{x_{i}\in A_{i}}\rho^{0}_{i}\left(x_{i},\mathbf{x}^{*}_{-i}\left(0\right)\right)
\end{equation}
Hence
\begin{equation}
\rho^{0}_{i}\left(\bar{x}_{i},\mathbf{x}^{*}_{-i}\left(0\right)\right)=\rho^{0}_{i}\left(x^{*}_{i}\left(0\right),\mathbf{x}^{*}_{-i}\left(0\right)\right)
\end{equation}
and the proof of the first point of the theorem under assumption H1) is completed. Assume H2) holds. The continuity of opportunity cost of uncertainty (i.e. the first point of the theorem) follows from Lemma \ref{Lemma::1}. If the first point holds true, the second point of the theorem is a consequence of Theorem \ref{Theorem::DeltaROEm} and the definition of $\epsilon$-Nash equilibrium.
\end{proof}

\bigskip

\section{- Algorithm for worst-case best-reply functions}\label{AlgorithmBR}

Consider a robust game, where the uncertainty sets are polyhedral, then the following algorithm can used to obtain the worst-case best-reply functions.

\noindent\rule{8cm}{1pt}

\textbf{Algorithm:}

\noindent\rule{8cm}{1pt}

\begin{itemize}
\item[1:] Define $\bar{W}^{\delta_{i}}_{i}$ as the set of corner points of $W^{\delta_{i}}_{i}$.
\medskip
\item[2:] Build the function $g^{\delta_{i}}_{i}:\bar{W}^{\delta_{i}}_{i}\times A_{-i}\rightarrow \mathbb{R}$, defined as follows
\begin{equation}
g_{i}^{\delta_{i}}\left(\boldsymbol{\alpha}_{i};\mathbf{x}_{-i}\right)=\arg\max_{x_{i}\in A_{i}} f_{i}\left(\boldsymbol{\alpha}_{i}|x_{i},\mathbf{x}_{-i}\right)
\end{equation}
\item[3:] Build the function $h^{\delta_{i}}_{i}:\bar{W}^{\delta_{i}}_{i}\times A_{-i}\rightarrow \mathbb{R}$, defined as follows
\begin{equation}\small
h_{i}^{\delta_{i}}\left(\boldsymbol{\alpha}_{i};\mathbf{x}_{-i}\right)=
\left\{
\begin{array}{ll}
1 \quad \quad & \text{if} \quad f_{i}\left(\boldsymbol{\alpha}_{i}|g^{\delta_{i}}_{i}\left(\boldsymbol{\alpha}_{i};\mathbf{x}_{-i}\right),\mathbf{x}_{-i}\right)=\displaystyle\min_{\boldsymbol{\alpha}_{j}\in W_{i}^{\delta_{i}}} f_{i}\left(\boldsymbol{\alpha}_{j}|g^{\delta_{i}}_{i}\left(\boldsymbol{\alpha}_{i};\mathbf{x}_{-i}\right),\mathbf{x}_{-i}\right)\quad \text{and}\quad g^{\delta_{i}}_{i}\left(\boldsymbol{\alpha}_{i};\mathbf{x}_{-i}\right)\in A_{i} \\
\\
0 & \text{otherwise}
\end{array}
\right .
\end{equation}
\item[4:] Player $i$'s worst-case best-reply function can be obtained as
\begin{equation}
R^{\delta_{i}}\left(\mathbf{x}_{-i}\right)=\sum_{\boldsymbol{\alpha}_{i}\in \bar{W}^{\delta_{i}}_{i}}g^{\delta_{i}}_{i}\left(\boldsymbol{\alpha}_{i};\mathbf{x}_{-i}\right)h^{\delta_{i}}_{i}\left(\boldsymbol{\alpha}_{i};\mathbf{x}_{-i}\right)-\sum_{\boldsymbol{\alpha}_{i}\in \bar{W}^{\delta_{i}}_{i}}\sum_{\substack{\boldsymbol{\alpha}_{j}\in \bar{W}^{\delta_{i}}_{i} \\ j>i }}g^{\delta_{i}}_{i}\left(\boldsymbol{\alpha}_{i};\mathbf{x}_{-i}\right)h^{\delta_{i}}_{i}\left(\boldsymbol{\alpha}_{i};\mathbf{x}_{-i}\right)h^{\delta_{i}}_{i}\left(\boldsymbol{\alpha}_{j};\mathbf{x}_{-i}\right)
\end{equation}
\end{itemize}

\noindent\rule{8cm}{1pt}

\section{- Proofs of the results on the robust duopoly model in Section \ref{Application}.}\label{AppendixB}
\begin{proof}[Proof of Proposition \ref{NENCDG}]
The action space of the nominal Cournot duopoly model can be partitioned in 4 regions
\begin{equation*}
\begin{array}{lllcclll}
\Omega _{1} &=&\left\{ \left( q_{1},q_{2}\right) |\frac{a}{\widehat{
\gamma }}>\text{ }q_{2}\geq 0\wedge \frac{a}{\widehat{\gamma }}>\text{ 
}q_{1}\geq 0\right\}  & \quad \quad & \Omega _{3} &=&\left\{ \left( q_{1},q_{2}\right) |\text{ }q_{2}\geq \frac{a}{\widehat{\gamma }}\wedge \frac{a}{\widehat{\gamma }}>\text{ }%
q_{1}\geq 0\right\}  \\
\Omega _{2} &=&\left\{ \left( q_{1},q_{2}\right) |\frac{a}{\widehat{
\gamma }}>\text{ }q_{2}\geq 0\wedge \text{ }q_{1}\geq \frac{a}{\widehat{\gamma }}\right\}  
& \quad \quad &
\Omega _{4} &=&\left\{ \left( q_{1},q_{2}\right) |q_{2}\geq \frac{a}{\widehat{\gamma }}0\wedge q_{1}\geq \frac{a}{\widehat{\gamma }}
\right\} \\
\end{array}
\end{equation*}
Moreover, by definition a Nash equilibrium $\left( q_{1}^{*},q_{2}^{*}\right)$ of the game solves
\begin{equation}
\left( q_{1}^{*},q_{2}^{*}\right) =\left( R^{0}\left( q_{2}^{*}\right)
,R^{0}\left( q_{1}^{*}\right) \right) 
\end{equation}
Hence, a sufficient condition for $\left( q_{1}^{*},q_{2}^{*}\right) \in
\Omega _{i}$, with $i\in \left\{ 1,2,3,4\right\} $, to be a Nash equilibrium is $\left( R^{0}\left(
q_{2}^{*}\right) ,R^{0}\left( q_{1}^{*}\right) \right) \in \Omega _{i}$.
Then, since $\forall \left( q_{1}^{*},q_{2}^{*}\right) \in $ $\Omega _{4}$ holds that $\left( R^{0}\left( q_{2}^{*}\right) ,R^{0}\left(
q_{1}^{*}\right) \right) =\left( 0,0\right) \notin \Omega _{4}$, Nash
equilibria do not exist in $\Omega _{4}$. At the same time, $\forall \left(
q_{1}^{*},q_{2}^{*}\right) \in $ $\Omega _{3}$ the condition $\left( R^{0}\left(
q_{2}^{*}\right) ,R^{0}\left( q_{1}^{*}\right) \right) =\left(
0,R^{0}\left( 0\right) \right) \in \Omega _{3}$ implies $R^{0}\left(
0\right) \geq \frac{a}{\widehat{\gamma }}$, or equivalently $\widehat{\gamma }>2\widehat{b} $, which violates Assumption \ref{Assumption1}. Thus, there are no Nash equilibria in $\Omega _{3}$. By symmetry, there are no Nash equilibria in $\Omega _{2}$. Finally, straightforward algebra shows that the Nash equilibrium in \eqref{NE:NominalGame} is the unique solution of the system $\left( q_{1},q_{2}\right) =\left( R^{0}\left( q_{2}\right)
,R^{0}\left( q_{1}\right) \right) $ in $ \Omega _{1}$ and it always exists under Assumption \ref{Assumption1}.
\end{proof}

\begin{proof}[Proof of Proposition \ref{ROERCDG}]
By hypothesis $\delta>0$. Then, the action space of the robust Cournot duopoly model in \eqref{G} can be partitioned in 16 regions.
\begin{equation*}\footnotesize
\begin{array}{lllcclll}
\Omega _{1} &=&\left\{ \left( q_{1},q_{2}\right) |\underline{q}\left(\delta\right)>\text{ }%
q_{1}\geq 0\wedge \underline{q}\left(\delta\right)>\text{ }q_{2}\geq 0\right\}  & \quad & \Omega _{9} &=&\left\{ \left( q_{1},q_{2}\right) |q ^{M}\left(\delta\right)>\text{ }%
q_{1}\geq \overline{q}\left(\delta\right)\wedge \underline{q}\left(\delta\right)>\text{ }q_{2}\geq
0\right\}\\
\Omega _{2} &=&\left\{ \left( q_{1},q_{2}\right) |\underline{q}\left(\delta\right)>\text{ }%
q_{1}\geq 0\wedge \overline{q}\left(\delta\right)>\text{ }q_{2}\geq \underline{q}\left(\delta\right)\right\} & \quad & \Omega _{10} &=&\left\{ \left( q_{1},q_{2}\right) |q^{M}\left(\delta\right)>\text{ }%
q_{1}\geq \overline{q}\left(\delta\right)\wedge \overline{q}\left(\delta\right)>\text{ }q_{2}\geq 
\underline{q}\left(\delta\right)\right\} \\
\Omega _{3} &=&\left\{ \left( q_{1},q_{2}\right) |\text{ }\underline{q}\left(\delta\right)>%
\text{ }q_{1}\geq 0\wedge q ^{M}\left(\delta\right)>\text{ }q_{2}\geq \overline{q}\left(\delta\right)\right\}  & \quad &  \Omega _{11} &=&\left\{ \left( q_{1},q_{2}\right) |q^{M}\left(\delta\right)>\text{ }%
q_{1}\geq \overline{q}\left(\delta\right)\wedge q^{M}\left(\delta\right)>\text{ }q_{2}\geq \overline{q}\left(\delta\right)\right\} \\
\Omega _{4} &=&\left\{ \left( q_{1},q_{2}\right) |\underline{q}\left(\delta\right)>\text{ }%
q_{1}\geq 0\wedge \text{ }q_{2}\geq q ^{M}\left(\delta\right)\right\}  & \quad & \Omega _{12} &=&\left\{ \left( q_{1},q_{2}\right) |q^{M}\left(\delta\right)>\text{ }%
q_{1}\geq \overline{q}\left(\delta\right)\wedge q_{2}\geq q^{M}\left(\delta\right)\right\} \\
\Omega _{5} &=&\left\{ \left( q_{1},q_{2}\right) |\overline{q}\left(\delta\right)>\text{ }%
q_{1}\geq \underline{q}\left(\delta\right)\wedge \underline{q}\left(\delta\right)>\text{ }q_{2}\geq
0\right\}  & \quad & \Omega _{13} &=&\left\{ \left( q_{1},q_{2}\right) |\text{ }q_{1}\geq
q^{M}\left(\delta\right)\wedge \underline{q}\left(\delta\right)>\text{ }q_{2}\geq 0\right\} \\
\Omega _{6} &=&\left\{ \left( q_{1},q_{2}\right) |\overline{q}\left(\delta\right)>\text{ }%
q_{1}\geq \underline{q}\left(\delta\right)\wedge \overline{q}\left(\delta\right)>\text{ }q_{2}\geq 
\underline{q}\left(\delta\right)\right\}  & \quad &  \Omega _{14} &=&\left\{ \left( q_{1},q_{2}\right) |q_{1}\geq q^{M}\left(\delta\right)\wedge 
\overline{q}\left(\delta\right)>\text{ }q_{2}\geq \underline{q}\left(\delta\right)\right\} \\
\Omega _{7} &=&\left\{ \left( q_{1},q_{2}\right) |\overline{q}\left(\delta\right)>\text{ }%
q_{1}\geq \underline{q}\left(\delta\right)\wedge q ^{M}\left(\delta\right)>\text{ }q_{2}\geq \overline{q}\left(\delta\right)\right\}  & \quad &  \Omega _{15} &=&\left\{ \left( q_{1},q_{2}\right) |q_{1}\geq q^{M}\left(\delta\right)\wedge
q^{M}\left(\delta\right)>\text{ }q_{2}\geq \overline{q}\left(\delta\right)\right\} \\
\Omega _{8} &=&\left\{ \left( q_{1},q_{2}\right) |\overline{q}\left(\delta\right)>\text{ }%
q_{1}\geq \underline{q}\left(\delta\right)\wedge \text{ }q_{2}\geq q ^{M}\left(\delta\right)\right\}  & \quad & \Omega _{16} &=&\left\{ \left( q_{1},q_{2}\right) |q_{1}\geq q^{M}\left(\delta\right)\wedge
q_{2}\geq q^{M}\left(\delta\right)\right\}
\end{array}
\end{equation*}
According to definition \eqref{defCournotROE}, $\left( q_{1}^{*},q_{2}^{*}\right)$ is a robust-optimization equilibrium of our Cournot duopoly model if and only if
\begin{equation}\label{CournotROEcond}
\left( q_{1}^{*},q_{2}^{*}\right) =\left( R^{\delta}\left( q_{2}^{*}\right)
,R^{\delta}\left( q_{1}^{*}\right) \right) 
\end{equation}
Then, a necessary condition for $\left( q_{1}^{*},q_{2}^{*}\right) \in
\Omega _{i}$, with $i\in \left\{ 1,2,...,16\right\} $, is $\left(
R^{\delta}\left( q_{2}^{*}\right) ,R^{\delta}\left( q_{1}^{*}\right) \right) \in
\Omega _{i}$. It follows that:
\begin{enumerate}
\item[-] There are no robust-optimization equilibria in $\Omega _{16}$. Indeed $\forall \left( q_{1},q_{2}\right) \in $ 
$\Omega _{16}$, we have that $\left( R^{\delta}\left( q_{2}\right)
,R^{\delta}\left( q_{1}\right) \right) =\left( 0,0\right) \notin \Omega
_{16} $;
\item[-] There are no robust-optimization equilibria in $\Omega _{15}\cup \Omega _{14}$. Indeed $\forall \left( q_{1},q_{2}\right) \in \Omega _{15}\cup \Omega _{14}
$ we have $\left( R^{\delta}\left( q_{2}\right) ,R^{\delta}\left( q_{1}\right) \right)
=\left( R^{\delta}\left( q_{2}\right) ,0\right) \notin $ $\Omega _{15}\cup \Omega
_{14}$. By symmetric arguments, the same holds true in $\Omega _{8}\cup \Omega _{12}$;
\item[-] There are no robust-optimization equilibria in $\Omega _{4}$. Indeed $\forall \left( q_{1}^{\ast },q_{2}^{\ast }\right) \in \Omega _{4}$
the condition $\left( R^{\delta}\left( q_{2}^{\ast }\right) ,R^{\delta}\left(
q_{1}^{\ast }\right) \right) =\left( 0,R^{\delta}\left( 0\right) \right) \in \Omega _{4}$ implies $R^{\delta}\left( 0\right) \geq q^{M}\left(\delta\right)$, or
equivalently $\underline{\gamma }\left( \delta\right) >2\overline{b}\left( \delta\right)$ which contradicts Assumption \ref{Assumption1}. By symmetric arguments, the same holds true in $\Omega _{13}$;
\item[-] There are no robust-optimization equilibria in $\Omega _{9}$. Indeed, solving \eqref{CournotROEcond} in $\Omega _{9}$ we obtain
\begin{equation}
q_{1}^{*}=\frac{ \left(2 
\underline{b}\left( \delta\right) -\underline{\gamma }\left( \delta\right)\right)a}{ 4\overline{b}\left( \delta\right) \underline{b}\left( \delta\right) -\overline{\gamma }\left( \delta\right) \underline{\gamma }
\left( \delta\right)} \quad \text{and} \quad 
q_{2}^{*}=\frac{\left( 2\overline{b}\left( \delta\right) -\overline{\gamma }\left( \delta\right) \right) a }{4\overline{b}\left( \delta\right) \underline{b}\left( \delta\right) -\overline{\gamma}\left(\delta\right) \underline{\gamma}\left( \delta\right) } 
\end{equation}
imposing $\left(q_{1}^{*},q_{2}^{*}\right)\in\Omega_{9}$ is equivalent to $q ^{M}\left(\delta\right) > q_{1}^{*}\geq \overline{q}\left(\delta\right)$ and $\underline{q}\left(\delta\right) >q_{2}^{*}>0$, which implies
\begin{equation}
\left( 2\underline{b}\left( \delta \right) -\underline{\gamma }\left( \delta
\right) \right) \left( \overline{\gamma }  -\underline{\gamma }  \right) >\left( 2\overline{b}\left( \delta
\right) -\overline{\gamma }\left( \delta \right) \right) \left( \overline{b}  -\underline{b}  \right)  \quad \text{and}\quad
\left( 2\underline{b}\left( \delta \right) -\underline{\gamma }\left( \delta
\right) \right) \left( \overline{b}  -\underline{b}\right) > \left( 2\overline{b}\left( \delta \right) -\overline{\gamma }\left( \delta \right) \right) \left( \overline{\gamma } -\underline{\gamma }  \right) 
\end{equation}
where one condition contradicts the other one. By symmetric arguments, there are not robust-optimization equilibria in $\Omega _{3}$.
\item[-] There are no robust-optimization equilibria in $\Omega _{2}$. Indeed, solving \eqref{CournotROEcond} in $\Omega _{2}$ we obtain
\begin{equation}
q_{1}^{*}=\frac{\left( \overline{\gamma } -
\underline{\gamma }  \right) a }{2
\overline{b}\left( \delta\right) \left( \overline{b}  -\underline{b} \right) +
\underline{\gamma }\left( \delta\right) \left( \overline{\gamma } -\underline{\gamma }  \right) }
\quad \text{and}\quad 
q_{2}^{*}=\frac{\left( \overline{b} -
\underline{b} \right) a }{2\overline{b}\left( \delta\right) \left( \overline{b}-\underline{b}  \right) +%
\underline{\gamma }\left( \delta\right) \left( \overline{\gamma }  -\underline{\gamma }  \right) }
\end{equation}
imposing $\left(q_{1}^{*},q_{2}^{*}\right)\in\Omega_{2}$ is equivalent to have $\underline{q}\left(\delta\right) >q_{1}^{*}>0$ and $\overline{q}\left(\delta\right)  >q_{2}^{*}\geq \underline{q}\left(\delta\right)$, which implies $\overline{b} -\underline{b} >\overline{\gamma } -\underline{\gamma } $ to have $q_{2}^{*}>q_{1}^{*}$ and $2\overline{b}\left( \delta\right)<\underline{\gamma}\left( \delta\right)$ to have $q_{2}^{*}\geq \underline{q}\left( \delta\right)$. However, the latter condition violates Assumption  \ref{Assumption1}, in fact $2\widehat{b}\left( \delta\right)>\widehat{\gamma}\left( \delta\right)$ implies $2\overline{b}\left( \delta\right)\left(>2\widehat{b}\left( \delta\right)>\widehat{\gamma}\left( \delta\right)\right)>\underline{\gamma}\left( \delta\right)$. By symmetric arguments, there are no robust-optimization equilibria in $\Omega _{5}$.
\end{enumerate}
In the reminder partitions of the state space of the game, robust-optimization equilibria exist under certain conditions:
\begin{itemize} 
\item Solving \eqref{CournotROEcond} in $\Omega _{1}$, we obtain $q_{1}^{*}=q_{2}^{*}=q^{ROE_{1}}$. Imposing $\left(q^{ROE_{1}},q^{ROE_{1}}\right)\in \Omega_{1}$ is equivalent to $\underline{q}\left(\delta\right)>q^{ROE_{1}}>0$ and it is satisfied if and only if $\overline{b}  -\underline{b}  >\overline{\gamma } -\underline{\gamma } $. Then $\left(q^{ROE_{1}},q^{ROE_{1}}\right)$ is a robust-optimization equilibrium if and only if the latter condition is satisfied.
\item According to \eqref{CournotROEcond}, a robust-optimization equilibrium in $\Omega _{6}$ implies
\begin{equation}
q^{*}_{1}=\frac{\overline{\gamma} -\underline{\gamma} }{\overline{b} -\underline{b} }q^{*}_{2} \quad \text{and}
\quad 
q^{*}_{2}=\frac{\overline{\gamma} -\underline{\gamma}
 }{\overline{b}  -
\underline{b} }q^{*}_{1}
\end{equation}
Then, $\overline{b} -\underline{b} =\overline{\gamma}  -\underline{\gamma} $ and hence a continuum of robust-optimization equilibria exists and is given by
\begin{equation}\label{ROE6}
\left\{\left(q_{1}^{*},q_{2}^{*}\right) \ | \ \overline{q}\left( \delta\right) >\text{ }q_{1}^{*}\geq \underline{q}\left( \delta\right) \wedge \overline{q}\left( \delta\right) >\text{ }q_{2}^{*}\geq \underline{q}\left( \delta\right) \right\} 
\end{equation}
where \eqref{ROE6} is obtained imposing the feasibility condition $\left(q_{1}^{*},q_{2}^{*}\right) \in\Omega_{6}$, or there are no robust-optimization equilibria.
\item Solving \eqref{CournotROEcond} in $\Omega _{11}$ under condition $\overline{\gamma}\left( \delta\right)\neq 2\underline{b}\left(\delta\right)$, we obtain $q_{1}^{*}= q_{2}^{*}=q^{ROE_{2}}$. 
Imposing $\left(q^{ROE_{2}},q^{ROE_{2}}\right)\in \Omega_{11}$, we have that it is a robust-optimization equilibrium if and only if $\overline{b} -\underline{b} <\overline{\gamma}  -\underline{\gamma} $. For $\overline{\gamma}\left( \delta\right)= 2\underline{b}\left(\delta\right)$, solving \eqref{CournotROEcond} in $\Omega_{11}$ we obtain $q^{*}_{1}=\frac{a }{2\underline{b}\left( \delta\right) }-q^{*}_{2}$. 
Therefore, $\left(q,\frac{a-c}{2\underline{b}\left( \delta\right) }-q\right)\in\Omega_{11}$ are robust-optimization equilibria.
\item Solving \eqref{defCournotROE} in $\Omega _{7}$, we obtain $q_{1}^{*}=q^{ROE_{3}}$ and $q_{2}^{*}=q^{ROE_{4}}$. Imposing $\left(q^{ROE_{3}},q^{ROE_{4}}\right)\in \Omega_{7}$, we have $
q^{M}\left(\delta\right) >q^{ROE_{4}}\geq \overline{q}\left(\delta\right)$ and $
\overline{q}\left(\delta\right)>q^{ROE_{3}}\geq \underline{q}\left(\delta\right)$.
The first condition is equivalent to $\overline{\gamma }
  -\underline{\gamma }  >\overline{b}  -\underline{b} $, which implies that the second reduces to $\overline{\gamma}\left( \delta \right)>2\underline{b}\left( \delta \right)$. Then $\left(
q^{ROE_{3}},q^{ROE_{4}}\right) $ is a robust-optimization equilibrium if and only if the latter two conditions are both satisfied. By symmetric arguments, $\left(q^{ROE_{4}},q^{ROE_{3}}\right)\in\Omega _{10}$ is a robust-optimization equilibrium if and only if the latter two conditions hold.
\end{itemize}
This completes the proof.
\end{proof}

\begin{proof}[Proof of Proposition \ref{AsROEProfits}]
By assumption $\left(q^{ROE_{3}},q^{ROE_{4}}\right)$ is a robust-optimization equilibrium of the game, then $\overline{\gamma}-\underline{\gamma}>\overline{b}-\underline{b}$ and $\delta>\delta^{*}$ by Theorem \ref{ROERCDG}. The first condition implies $q^{ROE_{4}}>q^{ROE_{3}}$, while the second condition is equivalent to $\overline{\gamma}\left(\delta\right)>2\underline{b}\left(\delta\right)$ and implies
\begin{equation}
\footnotesize
\rho^{\delta}\left(q^{ROE_{3}},q^{ROE_{4}}\right)-\rho^{\delta}\left(q^{ROE_{4}},q^{ROE_{3}}\right)= a\left(1-
\frac{\underline{b}\left( \delta\right)\left( \overline{b}\left(\delta\right) -\underline{b}\left( \delta\right) \right) +\underline{b}\left( \delta\right)\left( \overline{\gamma }\left(\delta\right) -\underline{\gamma }\left( \delta\right) \right)}{2\underline{b}\left( \delta\right) \left( \overline{b}\left(
\delta\right) -\underline{b}\left( \delta\right) \right) +\overline{\gamma }\left( \delta\right) \left( \overline{\gamma }\left(\delta\right) -\underline{\gamma }\left( \delta\right) \right) }\right) \left(q^{ROE_{3}}-q^{ROE_{4}}\right)<0
\end{equation}
The symmetry of the game completes the proof.
\end{proof}

\begin{proof}[Proof of Proposition \ref{AsROEProfitsC}]
The existence of the three robust-optimization equilibria implies $\overline{\gamma}-\underline{\gamma}>\overline{b}-\underline{b}$ and $\delta\in\left(\delta^*,1\right]$ (or equivalently $\overline{\gamma}\left(\delta\right)>2\underline{b}\left(\delta\right)$), from which it follows that $q^{ROE_{4}}>q^{ROE_{2}}>q^{ROE_{3}}$. Moreover, we have
\begin{equation}
\rho^{\delta}\left(q^{ROE_{2}},q^{ROE_{2}}\right) = \underline{b} \left(q^{ROE_{2}}\right)^2\quad \quad \quad \quad
 \rho^{\delta}\left(q^{ROE_{3}},q^{ROE_{4}}\right)=\underline{b} \left(q^{ROE_{3}}\right)^2
\end{equation}
and
\begin{equation}
\rho^{\delta}\left(q^{ROE_{4}},q^{ROE_{3}}\right) =\frac{
\left(2\underline{b}-\overline{\gamma}\right)
\left(\overline{b}-\underline{b}\right)
+\left(\overline{\gamma}-\underline{b}\right)   \left(\overline{\gamma}-\underline{\gamma}\right)}{\overline{\gamma}-\underline{\gamma}} \left(q^{ROE_{4}}\right)^2
\end{equation}
for all $\delta\in\left[\delta^{*},1\right]$. Therefore,
\begin{equation}
\rho^{\delta}\left(q^{ROE_{2}},q^{ROE_{2}}\right) -\rho^{\delta}\left(q^{ROE_{3}},q^{ROE_{4}}\right) = \underline{b}\left( q^{ROE_{2}}   - q^{ROE_{3}} \right)\left(  q^{ROE_{2}}   + q^{ROE_{3}} \right)>0
\end{equation}
and
\begin{equation}
\rho^{\delta}\left(q^{ROE_{4}},q^{ROE_{3}}\right) -\rho^{\delta}\left(q^{ROE_{2}},q^{ROE_{2}}\right) > \underline{b}\left(q^{ROE_{4}}- q^{ROE_{2}}  \right)\left(q^{ROE_{4}} + q^{ROE_{2}}  \right)>0
\end{equation}
which proves the Proposition.
\end{proof}

\bibliographystyle{spbasic}
\bibliography{binary_bib}

\end{document}